\documentclass[a4paper,UKenglish,cleveref,autoref]{article}

\usepackage[margin=1in]{geometry}
\usepackage{microtype}
\usepackage[utf8]{inputenc}
\usepackage{amsmath}
\usepackage{amsfonts}
\usepackage{graphicx}
\usepackage{subcaption}
\usepackage{authblk}

\graphicspath{{figs/}}

\usepackage{enumitem}
\usepackage{xspace}
\usepackage{cite}
\usepackage{hyperref}

\usepackage{thm-restate}

\newcommand{\Z}{\mathbb{Z}}
\newcommand{\N}{\mathbb{N}}

\newcommand{\domainBig}{\mathcal{G}_N}
\newcommand{\domain}{\mathcal{G}^+_N}
\newcommand{\death}{\mathcal{D}}
\newcommand{\Split}{\mathcal{S}}
\newcommand{\abs}[1]{\left| #1 \right|}

\usepackage[noabbrev,capitalise]{cleveref}
\usepackage{apxproof}
\newtheorem{theorem}{Theorem}
\newtheorem{definition}[theorem]{Definition}
\newtheorem{lemma}[theorem]{Lemma}

\usepackage{amsthm}

\theoremstyle{plain}
\newtheorem{cor}[theorem]{Corollary}
\newtheorem{observation}[theorem]{Observation}

\title{Distance bounds for high dimensional consistent digital rays and 2-D partially-consistent digital rays}

\author[1]{Man-Kwun Chiu\footnote{Partially supported by ERC StG 757609.}}
\affil[1]{Institut f\"ur Informatik, Freie Universit\"at Berlin, Berlin, Germany\\
	\texttt{chiumk@zedat.fu-berlin.de}}
\author[2]{Matias Korman\footnote{Supported by MEXT Kakenhi No.~17K12635 and the NSF award CCF-1422311.}}
\affil[2]{Department of Computer Science, Tufts University, Medford, MA, USA\\
	\texttt{matias.korman@tufts.edu}}
\author[3]{Martin Suderland}
\affil[3]{Faculty of Informatics, Universit\`{a} della Svizzera italiana, Lugano, Switzerland\\
	\texttt{martin.suderland@usi.ch}}
\author[4]{Takeshi Tokuyama\footnote{Supported by MEXT Kakenhi 17K19954 and 18H05291.}}
\affil[4]{Kwansei Gakuin University, Sanda, Japan\\
	
	\texttt{tokuyama@kwansei.ac.jp}}

\date{}

\begin{document}
	
	\maketitle
\begin{abstract}
	We consider the problem of digitalizing Euclidean segments. Specifically, we look for a constructive method to connect any two points in $\mathbb{Z}^d$. The construction must be {\em consistent} (that is, satisfy the natural extension of the Euclidean axioms) while resembling them as much as possible. Previous work has shown asymptotically tight results in two dimensions with $\Theta(\log N)$ error, where resemblance between segments is measured with the Hausdorff distance, and $N$ is the $L_1$ distance between the two points. This construction was considered tight because of a $\Omega(\log N)$ lower bound that applies to any consistent construction in $\mathbb{Z}^2$. 
	
	In this paper we observe that the lower bound does not directly extend to higher dimensions. We give an alternative argument showing that any consistent construction in $d$ dimensions must have $\Omega(\log^{1/(d-1)} N)$ error. We tie the error of a consistent construction in high dimensions to the error of similar {\em weak} constructions in two dimensions (constructions for which some points need not satisfy all the axioms). This not only opens the possibility for having constructions with $o(\log N)$ error in high dimensions, but also opens up an interesting line of research in the tradeoff between the number of axiom violations and the error of the construction. In order to show our lower bound, we also consider a colored variation of the concept of discrepancy of a set of points that we find of independent interest.
\end{abstract}	
	
\section{Introduction}
Euclidean line segments are one of the most fundamental objects of geometry. Although often loosely referred to as {\em the shortest path connecting the endpoints}, segments have a clear and unique axiomatic definition out of which many interesting properties follow. For example, it is well-known that the intersection of two segments is always a segment (that could possibly degenerate to a point or even become empty). The definition of other mathematical concepts heavily depends on the definition of segments
(e.g., we say that a certain region $P$ of the space is {\em convex} if for any two points $p,q\in P$, the line segment defined by $p$ and $q$ is in $P$). 

The definition of segment works very well in a Euclidean or similar spaces with infinite precision. Digital representation  (such as pixels in a screen) introduces imprecision. 
The most common approach used in practice is to somehow {\em round} the Euclidean segment into the digital space. The digital segments will look very similar to the Euclidean counterparts (that is, the {\em error} is very small).
However, 
we cannot guarantee the useful properties and concepts that follow from the axiomatic definition of Euclidean segment (see \cref{fig:rounding}).

In the aspect of the consistency of digital segments, we look for a deterministic method to construct digital segments in a way that $(i)$ the analogous of Euclidean axioms are satisfied and $(ii)$ the digital segments resemble the Euclidean ones as much as possible. 

\begin{figure}[h]
	\centering
	\includegraphics[page=1, scale=1]{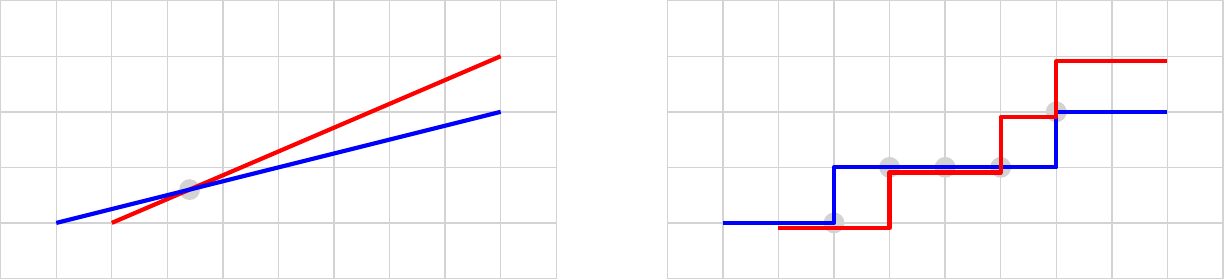}
	\caption{Left: Two Euclidean line segments that intersect in a point.
		Right: Rounding produces polylines that intersect in three disconnected components.}
	\label{fig:rounding}
\end{figure}

\subsection*{Preliminaries}
Our aim is to construct a digital path $dig(p,q)$ for any two points $p,q \in \Z^d$. Ideally, we want $dig$ to be 
defined for any pairs of points in $\Z^d$ (full list of requirements is described below), but sometimes we consider the case in which $dig$ is only defined for a subset of $\Z^d \times \Z^d$.

\begin{definition} For any $S\subseteq \Z^d \times \Z^d$, let $DS(S)$ be a set of digital segments such that $dig(p,q)\in DS(S)$ for all $(p,q)\in S$. We say that $DS(S)$ forms a {\em partial set of consistent digital segments} on $S$ (partial CDS for short) if for every pair $(p,q)\in S$ it satisfies the following five axioms:
	\begin{itemize}
		\item[(S1)] Grid path property: $dig(p,q)$ is a path between $p$ and $q$ under the $2d$-neighbor topology\footnote{The $2d$-neighbor topology is the natural one that connects to your predecessor and successor in each dimension. Formally speaking, two points are connected if and only if their $|| \cdot ||_{1}$ distance is exactly one.}. 
		\item[(S2)] Symmetry property: if $(q,p) \in S$, $dig(p,q)=dig(q,p)$.
		\item[(S3)] Subsegment property: for any $r\in dig(p,q)$, $dig(p,r) \in DS(S)$ and $dig(p,r) \subseteq dig(p,q)$.
		\item[(S4)] Prolongation property: $\exists~ r \in \Z^d$ such that $dig(p,r)\in DS(S)$ and $dig(p,q) \subset dig(p,r)$.
		\item[(S5)] Monotonicity property: for all $i\leq d$ such that $p_i = q_i$, it holds that every point $r\in dig(p,q)$ satisfies $r_i = p_i = q_i$.
	\end{itemize}
\end{definition} 

These axioms give nice properties of digital segments analogous to Euclidean line segments. For example, (S1) and (S3) imply that the intersection of two digital segments is another segment (that could degenerate to a single point or an empty set). (S5) implies that the intersection of a segment with an axis-aligned halfspace is a segment (and connected by (S1)), and so on. 

A partial CDS for $S = \Z^d\times \Z^d$ is called a set of {\em consistent digital segments} (CDS for short). 
Although our final goal is to have such a construction that works for the case in which $S=\Z^d\times \Z^d$, in this paper we consider subsets of the form $S = \{o\}\times \Z^d$ (where $o$ is the origin or any fixed point in $\Z^d$). We say that a partial CDS on such a set is a {\em consistent digital ray} system (CDR for short), as it contains all segments (or rays) from $o$ to $\Z^d$. 

Another property that we want from partial CDS is that they visually resemble the Euclidean segments. The resemblance between the digital segment $dig(p,q)$ and the Euclidean counterpart $\overline{pq}$ is measured using the {\em Hausdorff} distance. The Hausdorff distance $H(A,B)$ of two objects $A$ and $B$ is defined by $ H(A, B) = \max \{ h(A,B), h(B, A) \}$, where $h (A,B)
= \max_{a \in A} \min_{b \in B} \delta(a,b)$, and $\delta(a,b)$ is the standard $|| \cdot ||_{\infty}$ $L$-infinity norm. 

Thus, the resemblance of a partial CDS on $S$ is simply defined as $\max_{(p,q)\in S} H(dig(p,q),\overline{pq})$ (that is, the biggest error created between a digital segment and its Euclidean counterpart). This value is simply referred to as the {\em error} of the partial CDS construction. We are interested to see how the error grows as we enlarge our focus of interest. Thus, we limit the domain to the case in which both points are in the $L_1$ ball of radius $N$  centered at the origin (i.e.  $\domainBig=\Z^{d} \cap B_1(o, N)$). Rather than looking for the exact function, we are interested in the asymptotic behavior of the error as a function of $N$. For simplicity, we will actually restrict ourselves to the positive orthant $\domain=\domainBig \cap_i H_i$, where $H_i=\{p\in \Z^d \colon p_i \geq 0\}$ and $p_i$ is the $i$-th coordinate of $p$ (the results extend to other orthants by symmetry).

\subsection*{Previous Work}
Research on the digital representation of line segments has been an active area of research for over half a century~\cite{KR04}. Many different approaches have been considered. Most common techniques look for methods that implicitly encode the properties we desire. For example, a popular approach is to consider a {\em dynamic} method to digitize line segments. In this setting, the way we transform a Euclidean segment into a digital one will depend on which other segments are present (and their specific coordinates). 
It is known that a grid of exponential size is needed if we want to preserve the combinatorial types~\cite{DBLP:conf/stoc/GoodmanPS89}. Another workaround is known as  {\em snap rounding} that represents line segments by polygonal chains: Each segment is carefully rounded to avoid inconsistencies. Note that both of these ideas implicitly keep the error small while making sure that the intersection of two digital segments is a connected component. Although they work well in practice, they have the drawback that they cannot be used to define objects that are based on digital segments (such as digital starshapes or convex region).

The first paper to explicitly look for an axiomatic approach was in 1987 by Luby~\cite{luby}: in his work  he introduced the concept of CDS (under the name of {\em smooth geometries}) and gave a method to construct CDSs in $\Z^2$ based on a characterization of CDRs in $\Z^2$: any CDR can be uniquely identified by four total orders of the integers (and vice versa). By choosing a proper total order and using it for all points of $\Z^2$ we obtain a CDR with $O(\log N)$ error. H\r{a}stad\footnote{The lower bound was published by Luby, but credit given to H\r{a}stad (see Theorem~19 of~\cite{luby}).} gave a matching lower bound for any such construction. 
The lower bound is based on discrepancy theory~\cite{matousek}: any CDR is mapped to a sequence of real numbers in $[0,1)$ in a way that the error of the CDR is proportional to the discrepancy of the sequence (intuitively speaking, a measure on how well shuffled the numbers are).

These results were rediscovered by Chun {\em et al.}~\cite{cknt-cdg-09j} and Christ {\em et al.}~\cite{ChristJournal12}. They renewed interest in the topic and sparked other related research: Chowdhury and Gibson~\cite{Chowdhury2015} gave necessary and sufficient conditions for a collection of CDRs to form a CDS. In a companion paper, the same authors~\cite{Chowdhury2016} afterwards provided an alternative characterization together with a constructive algorithm; specifically, they gave an algorithm that, given a collection of segments in an $N \times N$ grid that satisfies the five axioms, computes a CDS that contains those segments. The algorithm runs in polynomial time of $N$. 

Unfortunately, most of these results only work on the digital plane. Out of the previously mentioned results, only the CDR construction of Chun {\em et al.}~\cite{cknt-cdg-09j} extends to three and higher dimensions. The construction has $O(\log N)$ error regardless of the dimension. 
Chun {\em et al.}~\cite{cknt-cdg-09j} 
also considered the case in which the monotonicity property (S5) is not preserved. They showed that if we remove (S5), we can obtain a CDR with $O(1)$ error in any dimension. Although the error is small, the resulting segments are far from what we would consider similar to the Euclidean segments (because they loop around many times). Recently, Chiu and Korman~\cite{ck-hdcds-18} showed that the problem in higher dimensions behaves very differently from the two dimensional case. Specifically, they studied how to extend the CDS construction of Christ {\em et al.}~\cite{ChristJournal12} and showed that it is very limiting in three (and higher) dimensions. We can use their method to get arbitrarily many CDRs (with $\Omega(\log N)$ error) and sometimes we can get a CDS. However, whenever the construction yields a CDS, it will have $\Omega(N)$ error.

Our interest in higher dimensions comes motivated by an application in {\em image segmentation}. Image segmentation is the act of separating an object from its background in an image (that is, determining which pixels are part of the background and which ones not). Chun {\em et al.}~\cite{cknt-cdg-09j} showed how to combine their CDR construction with the framework of Asano {\em et al.}~\cite{ACKT01} to segment two dimensional images. This idea has been extended to consider other shapes (see ~\cite{ckknt-acmwrdes-11} for a detailed list), but always two dimensional. The hope is that a high dimensional CDR with low error will produce more accurate segmentation algorithms. Although traditional images taken with a camera are two dimensional, images from a medical equipment such as those taken with an MRI machine can have three or even higher dimensions (say, when we want to track changes of a particular object along time).

\subsection*{Results and paper organization}
When approximating some geometric object, it often happens that higher dimensions create a larger error than in a lower dimension setting. Since the high dimensional setting contains a two dimensional subspace, it is common for lower bounds to extend to higher dimensions.
However, this is not true for the case of CDRs: although a three dimensional CDR contains two dimensional subspaces, those subspaces need not exactly be CDRs (and thus the $\Omega(\log N)$ lower bound does not directly hold). In this paper, we further  explain the reason and investigate the lower bound for the higher dimensional case.

The main reason why a subspace is not a CDR is because of the prolongation property (S4): we require that every segment is extendable, but has no constraints on the dimension in which it does so. In particular, a subspace of a high dimensional CDR need not be a CDR (see an example in \cref{fig:3d-CDR}). Subspaces of CDRs are what we call  {\em weak CDR}: it is a construction that almost always behaves like a CDR but some vertices may not satisfy the prolongation property (S4). Each vertex that does not extend is called an {\em inner leaf}. In this paper we study weak CDRs in two dimensions and the implications that they have for (proper) CDRs in higher dimensions.

\begin{figure}
	\centering
	\begin{subfigure}[b]{0.25\textwidth}
		\centering
		\includegraphics[scale=0.2,trim= 0 0 0 0, clip]{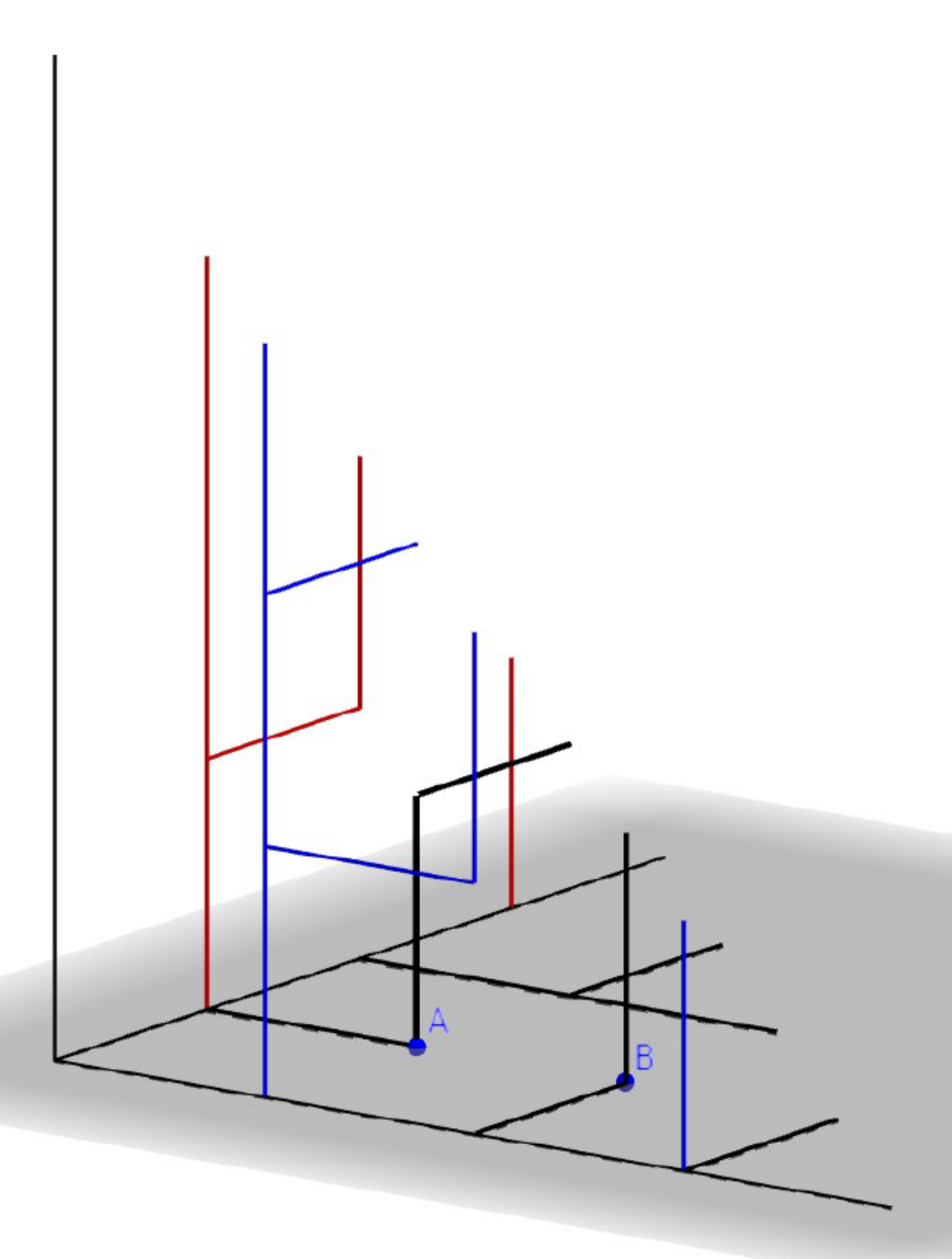}
	\end{subfigure}
	\hspace{1mm}
	\begin{subfigure}[b]{0.25\textwidth}
		\centering
		\includegraphics[scale=1.2,trim= 0 0 0 0, clip]{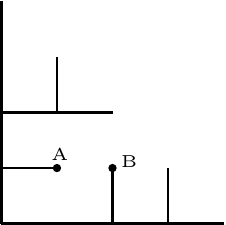}
		\vspace{1mm}
	\end{subfigure}
	\hspace{1mm}
	\begin{subfigure}[b]{0.25\textwidth}
		\centering
		\includegraphics[scale=0.4,trim= 10 0 0 0, clip]{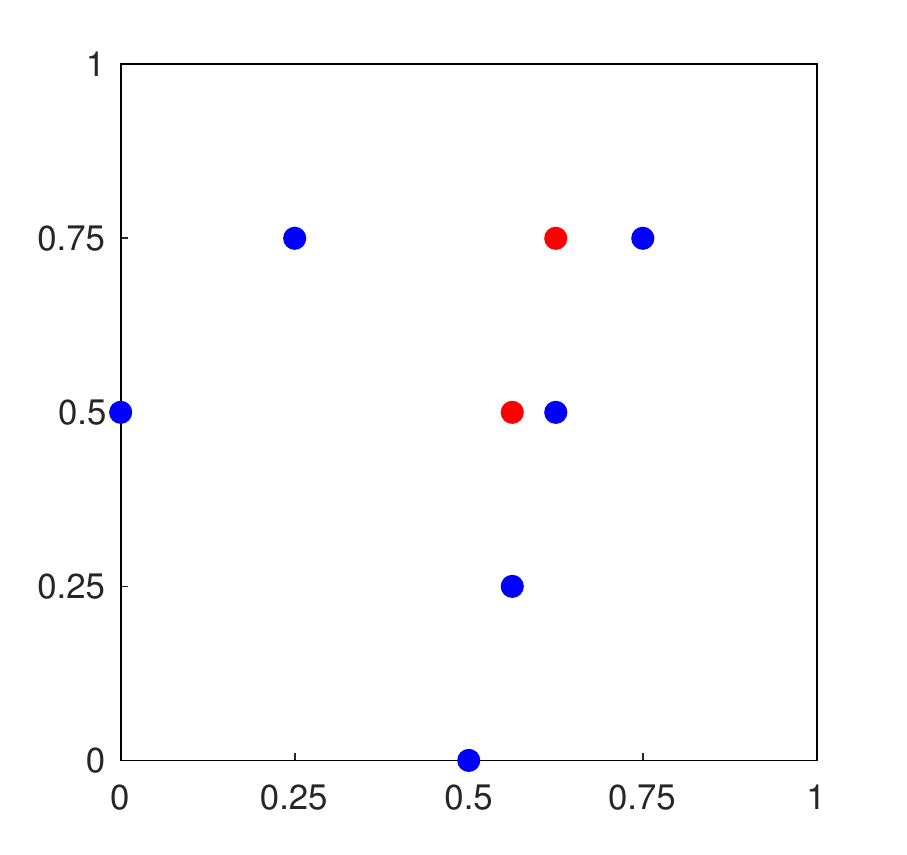}
	\end{subfigure}
	\caption{(left) A drawing of a CDR in $\domain \subset \Z^3$ for $N=4$. Notice that the CDR is a tree whose leaves are at the plane $x+y+z=N$. (middle) A cross section on the $xy$-plane of the same CDR. Observe that vertices A and B do not extend within the $xy$-plane. Thus, the subspace is a weak CDR (rather than a proper CDR). (right) A map of the weak CDR into a two-colored pointset. Regions with many blue points and few red correspond to portions of the CDR with high error.}
	\label{fig:3d-CDR}
\end{figure}

The new found properties of weak CDRs allow us to extend the two-dimensional lower bound to higher dimensions. H\r{a}stad's bound was based on a mapping from a (two-dimensional) CDR into a pointset in $[0,1)\subset \mathbb{R}$ and tied the error of the CDR to the discrepancy of the transformed pointset. Our lower bound uses an additional intermediate step: from any CDR we consider the weak CDR it generates in the $x_1x_2$-plane. We then map this weak CDR into a set of points in the unit square and {\em then} use discrepancy theory to obtain a lower bound for the weak CDR and eventually to the high dimensional CDR. Overall, we show a very strong link between the three spaces (CDR in high dimensions, weak CDR in the $x_1x_2$-plane and set of points created by our mapping). Along the paper we will analyze properties of each of the spaces, and see what implications it has for the other two. Specifically, we show the following: 
\begin{enumerate}[label=\roman*.]
	\item Because we now need to account for more general constructions (weak CDRs instead of proper CDRs), the mapping needs to be changed. Instead of creating points in the $[0,1)$ interval, in \cref{TransformationSection} we map into a two-colored pointset in $[0,1)\times [0,1)$.
	\item Similar to the two dimensional case, we can tie the error of the weak CDR to the discrepancy of the mapped pointset. First, we extend the discrepancy results~\cite{matousek} to our exact setting. 
	Let $R$ and $B$ be a set of red and blue points in the unit square, respectively. Let $m=|B|-|R|$ and assume $m>0$. For any set $P$ of points in the unit square and $x,y \in [0,1]$ let $P[x,y]$ be the number of points in $P \cap [0,x] \times [0,y]$. 
	For any two real numbers $0 \leq x,y\leq 1$ we define the discrepancy of $R$ and $B$ at $(x,y)$ as $D_{R,B}(x,y)=m x y - (B[x,y]-R[x,y]).$ The {\em discrepancy} of $R$ and $B$ is simply defined as $D^*_{R,B}=\max_{(x,y)\in[0,1]^2} | D_{R,B}(x,y) |$ (i.e., the highest discrepancy we can achieve among all possible rectangles). The discrepancy $D_{R,B}^*$ of a two-colored pointset is high if and only if there is an axis-aligned rectangle with the origin as corner in which the difference of the cardinalities is far from the expected difference.
	
	\begin{restatable}[Two colors discrepancy]{theorem}{discrep}
		\label{theo_discrep}
		For any set $R$ and $B$ of points such that $|B|>|R|$ it holds that 
		\begin{equation*} 
			D_{R,B}^* = \Omega\left(\frac{(|B|-|R|)\cdot \log(|B|+|R|)}{|B|+|R|}\right).
		\end{equation*}
	\end{restatable}		
	The proof is given in \cref{DiscrepancySection}.

	\item\label{relatwcdrpoint} With this new discrepancy result we obtain a trade-off between the error of any weak CDR 
	and the number of inner leaves (i.e., vertices that do not satisfy (S4)). When the weak CDR has zero inner leaves (and thus is a proper CDR) our bound matches the lower bound of H\r{a}stad. As the number of inner leaves increases, the lower bound decreases. In 
	\cref{TradeoffHausdorffErrorDeathVertices}
	we prove the following relationship.
	
	\begin{restatable}{theorem}{weakCDR}
		\label{theo_lower_2}
		For any $N \in \N$, any weak CDR defined on $\domain \subset \Z^2$ with $\kappa_2$ inner leaves between lines $x+y=\lceil N/2 \rceil$ and $x+y=N$ has $\Omega(\frac{N  \log N}{N+\kappa_2})$ error.
	\end{restatable}
	
	\item\label{lowerhigh} We then apply \cref{theo_lower_2}
	to obtain a lower bound for CDRs in $d$ dimensions: intuitively speaking, if the 2-D subspace has few inner leaves (say, $o(N\log N)$), then it will have $\omega(1)$ error. On the other hand, a weak CDR with many inner leaves in the 2-D subspace will cause too many points to extend to one of the remaining dimensions, and create large error 
	as well. This gives a lower bound of $\Omega(\log^{1/(d-1)}N)$ for any CDR construction in $d$ dimensions (see 
	\cref{3D-lower-bound}):
	
	\begin{restatable}{theorem}{lowerCDR}
		\label{thm:lowerCDR}
		Any CDR in $\Z^d$ has $\Omega( \log^{1/(d-1)} N)$ error.
	\end{restatable}
	
\end{enumerate}

Although we believe our analysis to be loose (especially in Theorem~\ref{thm:lowerCDR}), we are not certain that the existing CDR constructions with $O(\log N)$ error are tight either. 
In \cref{GreedyDiscrepancy,GreedyCDR},
we explore the possibility of having a CDR in high dimensions with $o(\log N)$ error (rather than directly looking at CDRs in high dimensions, we see what properties it would imply in the other two subspaces). Although we cannot explicitly find a construction with $o(\log N)$ error, we provide interesting insight on how further research can solve this question. 

In particular, in \cref{GreedyCDR} we give a weak CDR construction with $5/2$ error and $\Theta(N^2)$ inner leaves. 
In order to further reduce the number of inner leaves in weak CDRs with constant error we instead look at how to create a two-colored pointset with constant discrepancy in \cref{GreedyDiscrepancy}.
We show that it is not possible to have $o(N^2)$ red points in some pattern of the pointset with constant discrepancy,
which gives us a condition on any weak CDR with $o(N^2)$ inner leaves.

Further discussion on the implication of these results is given in \cref{sec_conclu}.



\section{Mapping a weak CDR into a pointset} \label{TransformationSection}

We start by showing how to transform a weak CDR in two dimensions into a two-colored pointset in $[0,1)^2$. 
Given any weak CDR, its restriction to $\domain$ forms a spanning tree $T$ of 
$\domain$ because of axioms (S1) and (S3). 
Although the tree is undirected, we see it as a directed graph (rooted tree) whose edges are oriented away from the origin (root). 
Then, (S5) implies that the parent of each vertex $(x,y)$ (except the root) is either $(x-1,y)$ or $(x,y-1)$.
For any edge $e=uv$ of $T$, where $u$ is the parent of $v$, we define $T(e)$ as the subtree of $T$ that is rooted at the child node $v$ of $e$. We slightly abuse the notation and use $T(v)$ to denote the subtree that is emanating from $v$ towards the leaves (that is, $T(v)=T(e)$). 

For any $n\leq N$ let $L_n$ be the points of $\domain$ whose sum of coordinates is $n$ (i.e., $L_n= \{(x,y) \in \domain \colon x+y = n\}$). 
We follow the usual terminology that we call a vertex of degree one a {\em leaf}.
We further consider two subcategories: we say that a leaf $v$ of $T$  is an {\em inner leaf} if it is not in $L_N$. All the vertices in $L_N$ are called {\em boundary leaves}. Note that, by properties of CDR, all vertices of $L_N$ are proper leaves (since any children should be in $L_{N+1}$, which is outside $\domain$). Further note that in a proper CDR there will be no inner leaves.   
A vertex $v$ of $T$ is a {\em split vertex} if it has degree three or it is the origin. Let $\Split$ be the set of split vertices and $\death$ the set of inner leaves. 

\subsection{Auxiliary function}\label{sec_auxiliary}
Before giving the transformation from a tree to a point set we first define an auxiliary function $M: \domain \rightarrow [0,1]$. 
For any $p \in L_N$ we set $M(p)=\frac{p_x}{p_x+p_y}$. For any subtree $T'(v)$ of $T$ we define two more functions inductively for $v \in L_{n}$ from $n = N$ to $0$ as follows:
\begin{align*}
	\max(T'(v)) & = \max_{p\in T'(v) \cap (\death \cup L_N)} M(p)\\
	\min(T'(v)) & = \min_{p\in T'(v) \cap (\death \cup L_N)} M(p),
\end{align*}
\noindent where $M(p)$ for $p \in \death$ is defined in the next paragraph.

For any inner leaf $\ell \in \death$, we know that the edges $e_1=(\ell_x-1,\ell_y+1)(\ell_x,\ell_y+1)$ and $e_2=(\ell_x+1,\ell_y-1)(\ell_x+1,\ell_y)$ must be present in $T$. Thus, we define $M(\ell)$ as $M(\ell)=\frac{\max(T(e_1))+\min(T(e_2))}{2}$. Intuitively speaking, we look at the leaves above and to the right of $\ell$, and assign a value that is in between the two of them (see \cref{fig:mapping}, left). The following statement shows that these values are sorted along $L_n$.

\begin{lemma}\label{lem_mapsorted}
	Let $T(u), T(v) \subset T$ be two subtrees of $T$ rooted at the vertices $u, v \in L_n$ (respectively) for some $n\leq N$  such that $u_x < v_x$. Then, it holds that $\max(T(u)) < \min(T(v))$.
\end{lemma}

\begin{proof}
	We prove this statement by induction on $n$ from $N$ to $1$. If both $u,v\in L_N$ then both $T(u)$ and $T(v)$ consist of a single vertex and the proof trivially follows. Now, assume that the statement is true for any two vertices $u', v' \in L_i$ for $i > n$. We need to show that the statement holds for any two vertices $u, v \in L_n$ such that $u_x < v_x$. 
	
	First observe that if we have two descendants $u'$ and $v'$ from $u$ and $v$ respectively such that $u',v'\in L_{n'}$ for some $n'>n$, then it holds that $u'_x < v'_x$. Indeed, this follows from the fact that when we embed $T$ in the natural way with edges drawn as straight segments, the result is a tree with no crossings. Thus, if $v'_x < u'_x$ happened for some descendants, then the two paths in $T$ from $u$ to $u'$ and from $v$ to $v'$ would either cross or form a cycle. Any of those two situations would contradict with the fact that $T$ is a weak CDR.
	
	Back to our original proof, consider the case in which neither $u$ nor $v$ are inner leaves. By the above argument we have that the $x$-coordinate of any child $u'\in L_{n+1}$ of $u$ must be smaller than any child $v'\in L_{n+1}$ of $v$. By induction, this implies that $\max(T(u')) < \min(T(v'))$ and thus $\max(T(u)) < \min(T(v))$.
	
	The cases in which $u$ or $v$ are inner leaves are similar: if $u$ is an inner leaf, we have $\max(T(u)) = M(u) =\frac{\max(T(u_1))+\min(T(u_2))}{2}$, where $u_1=(u_x,u_y+1)\in L_{n+1}$ and $u_2=(u_x+1,u_y)\in L_{n+1}$. By induction on $u_1$ and $u_2$ we have $\max(T(u_1)) < \min(T(u_2))$ and $\max(T(u)) < \min(T(u_2))$, thus we need to compare $\min(T(u_2))$ with any children of $v$. If $v$ is also an inner leaf, we can do a similar argument and have that $\max(T(v_1)) < \min(T(v))$ where $v_1=(v_x,v_y+1)$. 
	
	In general, given $u$, let $u'\in L_{n+1}$ be the child of $u$ with the largest $x$-coordinate (or $u'=u_2$ if $u$ is an inner leaf). Similarly, we define $v'$ as the child of $v$ with the smallest $x$-coordinate (or $v'=v_1$ if $v$ is an inner leaf). Again, by planarity of the natural embedding, we have that $u'_x \leq v'_x$ if at least one of $u,v$ is an inner leaf. In either case, we can use induction and get that $\max(T(u')) \leq \min(T(v'))$ which implies $\max(T(u)) < \max(T(u')) \leq \min(T(v')) \leq \min(T(v))$ (if $u$ is an inner leaf) or $\max(T(u)) \leq \max(T(u')) \leq \min(T(v')) < \min(T(v))$ (if $v$ is an inner leaf) completing the proof.
\end{proof}

For any subtree $T'$ of $T$, its {\em depth} is the longest possible length of a path from its root to any of its leaves. Any split vertex $s\in \Split$ has two branching edges $e_1$ and $e_2$, each defining a subtree. The subtree of higher depth is the {\em preferred} subtree of $s$ (in case of tie, we choose the tree emanating from $(s_x+1,s_y)$). For any point $p\in\domain$ we define a walk from $p$ to some leaf of $T$. If $p\in L_n$ has degree two, we follow the single edge to $L_{n+1}$. If $p \in \Split$, we follow the edge to the preferred subtree. This process is continued until we reach a leaf $\gamma(p)$. 

With this virtual walk we can define the function $M$ to all points $p\in\domain$ (not only leaves) of the domain as follows. If $p$ is neither a split nor a leaf, we define $M(p)$ as $M(p)=M(\gamma(p))$. For a split vertex $s$, let $s'$ be the child of $s$ that is not on the preferred subtree of $s$. Then, we define $M(s)$ as $M(s)=M(\gamma(s'))$. 

Intuitively speaking, from any vertex we always follow its only edge away from the root (if it has degree 2) or the preferred edge (if it has degree 3) until we reach a leaf. The only exception is if we start on a split vertex, in which case we do not follow the preferred edge at the first step. This exception is needed to make sure that the end points of the walk starting from split vertices are distinct.

\begin{lemma}\label{lem_bijection}
	For any split vertex $s\in \Split$, there exists a unique leaf $\ell \in\death \cup L_N$ such that $M(s)=M(\ell)$. And for any leaf $\ell \in \death \cup L_N \setminus \{(N,0)\}$, there exists a unique split vertex $s \in \Split$ such that $M(s)= M(\ell)$.
\end{lemma}

\begin{proof}
	By definition of the auxiliary function, two leaves do not have the same mapping. Thus, it remains to show that the walk of two different split vertices cannot end at the same leaf. Imagine doing the walk backwards: start at any leaf, walk towards the origin and stop as soon as you reach a split vertex by traversing its non-preferred edge. Since each split vertex has exactly two children, it follows that exactly one leaf will stop at each split vertex. The exceptional case is the leaf $(N,0)$, from which walking backwards to the origin is a horizontal path and the path does not contain any non-preferred edge. That is, in the inverse walk we follow preferred edges until we reach a non-preferred edge. This is equivalent to starting at a split vertex and follow the non-preferred edge once and continue with the preferred edges, which is the exact definition of our auxiliary function.
\end{proof}

\subsection{Transforming the tree into a pointset}

\begin{figure}
	\centering
	\includegraphics[page=3, scale=.65]{DetailedSplitAxesTicks} 
	\includegraphics[trim= 20 10 20 30,scale=0.43]{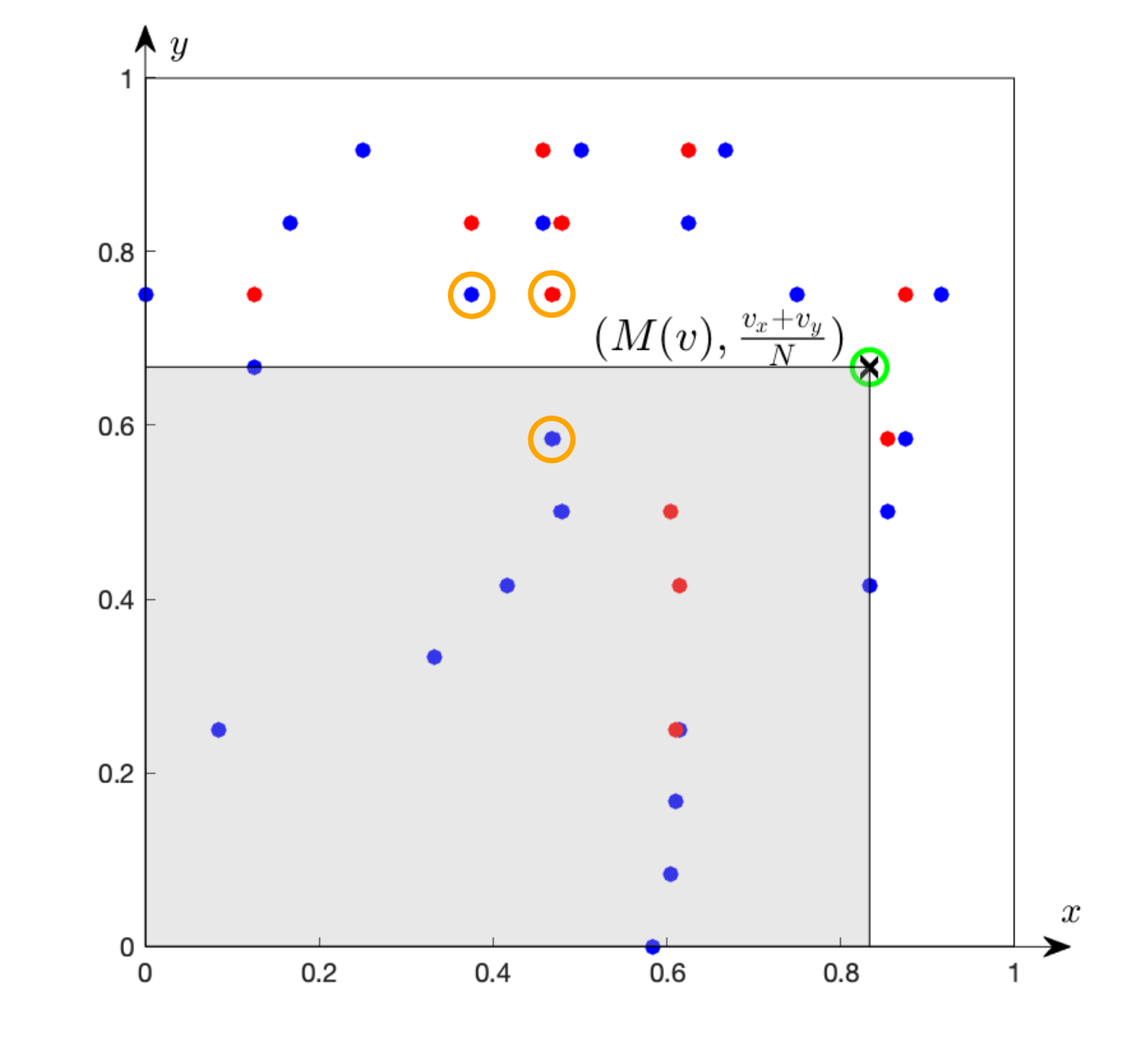}
	
	\caption{(left) A tree of a weak CDR and the value of the auxiliary function $M$ applied to all leaves of the tree. (right) The tree transformed into blue and red point sets. Two vertices of the same layer are mapped to points with the same y-coordinate and an inner leaf and its corresponding split vertex are mapped to points with the same x-coordinate (see the highlighted orange circles). (For Theorem~\ref{theo_interp}) The x-coordinate of $v = (6,2)$ (green circle) can be bounded in terms of the difference between blue and red point in the axis-aligned rectangle with corners $(0,0)$ and $\pi(v) = (M(v),\frac{v_x+v_y}{N}) =(\frac{10}{12},\frac{8}{12})$. The rectangle contains 11 blue points and 3 red ones.}
	
	\label{fig:mapping}
	\label{fig:intuitiontransformation}	
\end{figure}

With the auxiliary function $M$ we can define the mapping between a weak CDR into a bicolored pointset in the unit square. For any vertex $v=(v_x,v_y) \in \domain$ we define its {\em transformation} as $\pi(v)=(M(v), \frac{v_x+v_y}{N})$. 
Given any weak CDR, we look at the tree $T$ it defines in $\domain$. Each vertex $v \in\death$ creates a red point $\pi(v)$ and each split vertex $w\in\Split$ creates a blue point $\pi(w)$ (note that we do not transform the boundary leaves in $L_N$ into points). We define the {\em mapping} of $T$ as the union of the sets $R=\{\pi(v)\colon v\in \death\}$ and $B=\{\pi(v)\colon v\in \Split\}$ (see \cref{fig:mapping}, right). Note that the two sets depend on the tree $T$ (and thus $R=R(T)$ and $B=B(T)$). From now on we assume that $T$ is fixed, and thus we simplify the notation for ease of reading. 
For any set $P$ of points in the unit square and $x,y \in [0,1]$ let $P[x,y]$ be the number of points in $P \cap [0,x] \times [0,y]$. 

\begin{lemma}\label{lem_one_line}
	For any weak CDR $T$ in $\domain \subset \Z^2$ and $n <N$, the red and blue points on the line $y=n/N$ alternate in color starting and ending with a blue point. In particular, we have $B[1,n/N]-R[1,n/N]=n+1$.
\end{lemma}

\begin{proof}
	For the first statement we observe that only points that lie in $L_n$ will have $y$-coordinates equal to $n/N$. Moreover, since $L_{n+1}$ has one more vertex than $L_n$, each diagonal must have exactly one more split vertex than inner leaves. Indeed, Chun {\em et al.} showed that in proper CDRs each diagonal has exactly one split vertex (and of course, zero inner leaves).
	
	Now we need to show that split vertices and inner leaves appear alternatingly on the diagonal line. Consider two consecutive split vertices $u, v \in L_n$ such that $u_x < v_x$. By definition of split, the edges $e_u =(u_x, u_y)(u_x+1,u_y)$ and $e_v = (v_x, v_y)(v_x, v_y+1)$ are all in $T$. Observe that there are $v_x - u_x - 1$ vertices in $L_n$ and $v_x - u_x - 2$ vertices in $L_{n+1}$ between $e_u$ and $e_v$. Since two different vertices of $L_n$ cannot connect to the same vertex of $L_{n+1}$,
	one of them will not reach $L_{n+1}$. That vertex will be an inner leaf and will be between $u$ and $v$ as claimed. 
	
	
	That is, the blue pointset has one more point than the red pointset in each horizontal line $y = i/N$. Summing up the differences from $i=0$ to $n$, we get that in total there are $n+1$ additional blue points $p=(x,y)$ with $y\leq n/N$. 
\end{proof}

With the above observations we can now state the main relationship between the weak CDR and its mapped pointset.
For any vertex $v\in L_n$, its path to the origin splits the tree into two portions. Consider the portion of the tree up to $L_n$ that is above the path from $v$ to the origin. In $L_0$, the subtree contains a single vertex (the root) whereas at the diagonal $L_n$ contains $v_x+1$ vertices.  
Since the number of leaves grows with split vertices and shrinks with inner leaves, this means that in the portion of the tree that we are looking at, the difference between split vertices and inner leaves must be $v_x$, see \cref{fig:intuitiontransformation}. Note that if the two children of a split vertex (e.g., $(5,0)$ in \cref{fig:intuitiontransformation}) are not in the same portion, the number of leaves does not grow with that split vertex. However, these split vertices may be still contained in the rectangle that we consider in the mapped pointset. This is the reason why we do not have an equality in Theorem~\ref{theo_interp}.

\begin{theorem}\label{theo_interp}
	For any vertex $v \in \domain$ it holds that $B[M(v),\frac{v_x+v_y}{N}] - R[M(v),\frac{v_x+v_y}{N}] - 2\leq v_x \leq B[M(v),\frac{v_x+v_y}{N}] - R[M(v),\frac{v_x+v_y}{N}]$.
\end{theorem}

\begin{proof}
	We split the proof into two auxiliary lemmas.
	
	\begin{lemma}\label{emptyrectangle}
		Let $v \in L_n$ be a split vertex such that $v_x < n$. If $M(v) < M(\gamma(v))$ the rectangle $[M(v),M(\gamma(v))] \times \left[0,\frac{n-1}{N}\right]$ contains exactly one point, which is blue and has $M(\gamma(v))$ as x-coordinate. If $M(\gamma(v)) < M(v)$ the rectangle $[M(\gamma(v)),M(v)] \times \left[0,\frac{n-1}{N}\right]$ contains exactly one point, which is blue and has $M(\gamma(v))$ as x-coordinate. When  $v=(n,0) \in L_n$ the rectangle $[M(v),M(\gamma(v))] \times \left[0,\frac{n-1}{N}\right]$ is empty.
	\end{lemma}   
	\begin{proof}
		We first consider the case of $M(v) < M(\gamma(v))$.
		When we keep following from $v$ to the preferred subtree, we end up in a leaf, called $\ell$. By definition of $M$ we have $M(\ell) = M(\gamma(v))$. Since $v_x < n$ we have $M(\gamma(v)) \neq 1$. By Lemma~\ref{lem_bijection} there is a unique split vertex $s \in \Split$ such that $M(s)= M(\ell)$. This split vertex is below layer $L_n$ (indeed, we reach $L_n$ from $\ell$ by following only preferred edges and the inverse walk has to stop when we traverse a non-preferred edge of $s$) and therefore $s$ is transformed to a blue point in the rectangle.
		Now let $s'$ be a split vertex which is mapped to a blue point in the rectangle. We will show that $s' = s$. Let $\ell'$ be the unique leaf such that $M(\ell')= M(s')$. Consider first the case in which $\ell'$ is below layer $L_n$ (that is,  $\ell'_x+\ell'_y < n$). Then let $v'$ be the vertex on $dig(o,v)$ and $L_{\ell'_x+\ell'_y}$. If $\ell'_x < v'_x$ (resp. $v'_x < \ell'_x$) then Lemma~\ref{lem_mapsorted} implies that $M(\ell') < min(T_{v'}) \leq M(v)$ (resp. $M(\gamma(v)) \leq max(T_{v'}) < M(\ell')$). This would be a contradiction to $s'$ being mapped to a blue point in the rectangle. 
		
		It remains to consider the case in which $\ell'$ is above layer $L_n$. Define $\ell''$ to be the vertex on $dig(o, \ell')$ and $L_n$.  Lemma~\ref{lem_mapsorted} implies that $\ell'' = v$ (otherwise we have either $M(\ell') < M(v)$ or $M(\gamma(v))<  M(\ell')$ which would again be a contradiction). Recall that there is only one split vertex whose walk to its corresponding leaf through preferred subtrees passes through $v$. Hence $s' = s$ and there is exactly one blue point in the rectangle. 
		
		We now show that there cannot be any red point either. Indeed, recall that for every red point there is a blue point with the same $x$-coordinate and smaller $y$-coordinate because for each inner leaf $\ell$ there is a unique split vertex $s$ defined by the walk from $s$ to $\ell$ such that $M(\ell) = M(s)$. From the previous argument, we know that $s$ with $M(s) = M(\gamma(v))$ is mapped to the only one blue point in the rectangle and its corresponding leaf $\ell$ defined by the walk is above $L_n$. Hence, even if $\ell$ is an inner leaf, the mapped red point is not in the rectangle. 
		Moreover, there cannot be any other red point in the rectangle (since it would imply that the corresponding blue point would also be in and we already ruled out this case). 
		
		In the same way we can also prove that if $M(\gamma(v)) < M(v)$ the rectangle $[M(\gamma(v)),M(v)] \times \left[0,\frac{n-1}{N}\right]$ contains exactly one point, which is blue and has $M(\gamma(v))$ as x-coordinate.
		If $v_x = n$ then $\ell$ as defined above is the leaf $(N,0)$ and $M(\ell) = 1$.  Lemma~\ref{lem_bijection} implies that there is no split vertex $s$ with $M(s) = 1$. 
	\end{proof}
	

	\begin{lemma}\label{lem_interp}
		For any vertex $v \in \domain$ it holds that 
		\begin{equation}
			v_x - B\left[M(v),\frac{v_x+v_y-1}{N}\right] + R\left[M(v),\frac{v_x+v_y-1}{N}\right] + 1 \in \{0,1\}.
		\end{equation}
	\end{lemma}
	\begin{proof}
		We first prove by induction over $n$ that $\forall n \in \{0,1,...,N\}$ the following statement holds.
		\begin{equation} \label{claim}
			\left\{M(\gamma(p)) | p \in L_n\right\} = \left\{x \in [0,1]:  \abs{B \cap \{x\} \times \left[0,\frac{n-1}{N}\right]} - \abs{R \cap \{x\} \times \left[0,\frac{n-1}{N}\right]} = 1 \right\} \cup \{1\} .
		\end{equation}
		The quantity $|B \cap \{x\} \times [0,\frac{n-1}{N}]| - |R \cap \{x\} \times [0,\frac{n-1}{N}]|$ counts the difference between the number of blue points and red points on the vertical segment with $x$-coordinate $x$ and length $\frac{n-1}{N}$. Because of  Lemma~\ref{lem_bijection} we know that each split vertex shares the same value with a leaf in the auxiliary function $M$. If the leaf is an inner leaf, both blue (split) and red (inner) points lie on the same unit segment $\{x\} \times [0,1]$. Otherwise, there is only one blue point on $\{x\} \times [0,1]$ because $M(p)$ for $p \in L_N$ are all different.
		Hence the quantity $|B \cap \{x\} \times [0,\frac{n-1}{N}]| - |R \cap \{x\} \times [0,\frac{n-1}{N}]|$ can either be $0$ or $1$.
		
		The base case $n=0$ trivially holds. We have $ \{M(\gamma(p)) | p \in L_0\} = \{1\}$ and
		\begin{equation*}
			B \cap \{x\} \times \left[0,\frac{n-1}{N}\right] = R \cap \{x\} \times \left[0,\frac{n-1}{N}\right] = \emptyset
		\end{equation*}
		We assume that \cref{claim} holds for layer $L_n$ and we prove that it also holds for $L_{n+1}$.
		We distinguish 3 cases for any vertex $q$ in layer $L_n$. 
		\begin{itemize}
			\item If $q$ has degree 2 then $q$ and its child $r \in L_{n+1}$ are mapped by $M \circ \gamma$ to the same value. Moreover $q$ does not create any vertex in the set $B$ nor $R$. 
			\item If $q$ is an inner leaf, then the value $M(\gamma(q))$ will not appear in $\{M(\gamma(p)) | p \in L_{n+1}\}$ any more. The value $M(\gamma(q))$ also disappears in 
			\begin{equation*}
				\left\{x \in [0,1]: \left|B \cap \{x\} \times \left[0,\frac{n}{N}\right]\right| - \left|R \cap \{x\} \times \left[0,\frac{n}{N}\right]\right| = 1 \right\} \cup \{1\}.
			\end{equation*}
			because $q$ created a red point in $R$ with the coordinates $(M(\gamma(q)),\frac{n}{N}) = (M(q),\frac{n}{N})$.
			\item If $q$ is a split vertex, then the value $M(\gamma(q))$ will stay in $\{M(\gamma(p)) | p \in L_{n+1}\}$. Moreover \linebreak $\{M(\gamma(p)) | p \in L_{n+1}\}$ contains the additional value $M(q)$. The value $M(q)$ also appears in 
			\begin{equation*}
				\left\{x \in [0,1]:  \left|B \cap \{x\} \times \left[0,\frac{n}{N}\right]\right| - \left|R \cap \{x\} \times \left[0,\frac{n}{N}\right]\right| = 1 \right\} \cup \{1\}
			\end{equation*}
			because $q$ creates a blue point in $B$ with the coordinates $(M(q),\frac{n}{N})$.
		\end{itemize}
		
		Hence \cref{claim} holds. \newline 
		Let $v$ be a vertex in layer $L_n$, i.e. $n = v_x + v_y$. By Lemma~\ref{lem_mapsorted} we know that a vertex $u \in L_n$ with $u_x < v_x$ satisfies $M(\gamma(u)) < M(\gamma(v))$. By Lemma~\ref{lem_mapsorted} we also know that a vertex $w \in L_n$ with $v_x < w_x$ satisfies $M(\gamma(v)) < M(\gamma(w))$. Hence the number of vertices in layer $L_n$ with smaller $x$-coordinate than that of $v$ is exactly the number of vertices which are mapped by $M \circ \gamma$ to a smaller value than that of $v$.
		If $v_x < n$:
		\begin{align}
			v_x & = |\{u \in L_n| u_x < v_x\}| \stackrel{\text{Lemma~\ref{lem_mapsorted}}}{=} |\{u \in L_n| M(\gamma(u)) < M(\gamma(v))\}| \nonumber \\
			& = |\{u \in L_n| M(\gamma(u)) \leq M(\gamma(v))\}| -1 \nonumber \\
			& \stackrel{(\ref{claim})}{=} B\left[M(\gamma(v)),\frac{n-1}{N}\right] - R\left[M(\gamma(v)),\frac{n-1}{N}\right] - 1 \\
			& \stackrel{\text{Lemma~\ref{emptyrectangle}}}{=} \begin{cases}
				B[M(v),\frac{n-1}{N}] - R\left[M(v),\frac{n-1}{N}\right] -1 &\mbox{if } M(\gamma(v)) \leq M(v) \\
				B[M(v),\frac{n-1}{N}] - R\left[M(v),\frac{n-1}{N}\right] & \mbox{if } M(v) < M(\gamma(v))
			\end{cases} \label{Momega2M}
		\end{align}
		
		If $v_x = n$ then:
		\begin{align*}
			v_x & = |\{u \in L_n| M(\gamma(u)) \leq M(\gamma(v))\}| -1 \stackrel{(\ref{claim})}{=} B\left[M(\gamma(v)),\frac{n-1}{N}\right] - R\left[M(\gamma(v)),\frac{n-1}{N}\right] \nonumber \\
			& \stackrel{\text{Lemma~\ref{emptyrectangle}}}{=} B\left[M(v),\frac{n-1}{N}\right] - R\left[M(v),\frac{n-1}{N}\right]
		\end{align*}
	\end{proof}
	
	By Lemma~\ref{lem_one_line}, the red and blue points on the line $y= v_x+v_y$ alternate in color starting and ending with a blue point. Hence, any interval $[0, x]$ on the line $y=v_x+v_y$ contains at most one more blue points. Therefore, 
	$B[M(v),\frac{v_x+v_y}{N}] - R[M(v),\frac{v_x+v_y}{N}] - ( B[M(v),\frac{v_x+v_y-1}{N}] - R[M(v),\frac{v_x+v_y-1}{N}] )$ is at most one.
	Lemmas~\ref{lem_interp} and \ref{lem_one_line} directly imply \cref{theo_interp}. 
\end{proof}



\section{Bichromatic discrepancy} \label{DiscrepancySection}
Let $R$ and $B$ be a set of red and blue points in the unit square, respectively. Let $r=|R|$ and $b=|B|$, and further assume that $b>r$. Let $m=b-r$ (which is positive since $b>r$). For any set $P$ of points in the unit square and $x,y \in [0,1]$ let $P[x,y]$ be the number of points in $P \cap [0,x] \times [0,y]$. 

For any two sets $R$ and $B$ and real numbers $x,y\leq 1$ we define the discrepancy of $R$ and $B$ at $(x,y)$ as 
\begin{equation}\label{D_m}
	D_{R,B}(x,y)=(b-r)xy - (B[x,y]-R[x,y]).
\end{equation}

The {\em discrepancy} of $R$ and $B$ is simply defined as $D^*_{R,B}=\max_{(x,y)\in[0,1]^2} | D_{R,B}(x,y) |$ (i.e., the highest discrepancy we can achieve among all possible rectangles).

\discrep*

Note that if we set $R=\emptyset$ we get the classic two dimensional discrepancy result for which there are several proofs (see~\cite{matousek} for a detailed survey). In order to extend the bound for the case of $R\neq \emptyset$, we make minor changes to Schmidt's proof~\cite{Schmidt1972}. We start by using an auxiliary function $G$ (defined below) and combining it  with the trivial inequality

$$\int_{(x,y)\in[0,1]^2} D_{R,B}(x,y) G(x,y) dx dy \leq \max_{(x,y)\in[0,1]^2} | D_{R,B}(x,y) | \int_{(x,y)\in[0,1]^2} |G(x,y)| dx dy$$

to obtain

$$D^*_{R,B}=\max_{(x,y)\in[0,1]^2} |D_{R,B}(x,y)| \geq \frac{\int D_{R,B} G}{\int |G|}.$$

Note that for simplicity in the notation we removed the integration limits. Our definition of $G$ is identical to the one used by Schimdt: Let $m=\lceil \log_2 (b+r) \rceil +1$ and observe that, by definition of $m$ we have $2(b+r) \leq 2^m \leq 4(b+r)$. For any $j\in \{0, \ldots, m\}$ we define function $f_j: [0,1]^2 \rightarrow \{-1,0,1\}$ as follows: subdivide the unit square with $2^j$ equally spaced vertical lines and $2^{m-j}$ horizontal lines.

For any value of $j$ we subdivide the unit square into rectangles of area $2^{-m}$ (larger values of $j$ will result in thinner but wider rectangles). Let $A$ be a rectangle of subdivision associated to $f_j$. We define $f_j$ within the rectangle to be $0$ if $A$ contains any point of $R \cup B$. If $A$ does not have neither red nor blue points, we further subdivide it into four congruent quadrants. The function value of $f_j$ is equal to $1$ in the upper right and lower left quadrants, and $-1$ in upper left and lower right quadrants (see a visual representation of $f_j$ in~\cite{matousek}, page 173).

Then, we define $G$ as $G=(1+cf_0) (1+cf_1)\ldots (1+cf_m)-1$, where $c>0$ is a small constant (whose value will be chosen afterwards). Note that $G$ can also be expressed as $G=G_1 + \ldots G_m$, where 

$$G_k=c^k \sum_{0 \leq j_1 \leq \ldots \leq j_m \leq m} f_{j_1}f_{j_2}\ldots f_{j_k}.$$

Schmidt showed that $\int |G| \leq 2$ (regardless of the value of $m$). Thus, we now focus in giving an upper bound for $\int D_{R,B}G$.

\begin{lemma}\label{lem_lower}
	There exists a constant $c_1$ such that $\int D_{R,B}G_1 \geq c c_1\frac{b-r}{b+r} \log (b+r)$. 
\end{lemma}
\begin{proof}
	By definition of $G_1$ we have $\int D_{R,B}G_1 = c\sum_{j=0}^m \int D_{R,B}f_j$. Thus, it suffices to show that for any value of $j$ it holds that $\int D_{R,B}f_j \geq c' \frac{b-r}{b+r}$ (for some other constant $c'>0$).
	
	Recall that, when defining $f_j$, we subdivided the unit square into at least $2(b+r)$ rectangles. For the rectangles that contain at least one point of $R \cup B$, $f_j$ is set to zero, and thus they do not contribute to the integral. Since we have $b+r$ many points, we know that there must exist at least $b+r$ rectangles that do not contain any point of $R$ or $B$. Let $A$ be any such rectangle, and let $A_{SW},A_{NW},A_{SE},A_{NE}$ be the four subquadrants of $A$ (where the subindex refers to the cardinal position of the quadrant). Recall that $f_j$ is equal to $1$ for any point of $A_{SW} \cup A_{NE}$ and $-1$ for points of $A_{SE} \cup A_{NW}$.
	
	Let $\mathsf{w}$ and $\mathsf{h}$ be vectors defined by the horizontal and vertical sides of $A_{SW}$, respectively. Observe that their lengths are $2^{-j-1}$ and $2^{j-m-1}$, respectively. Then, we have 
	
	\begin{eqnarray*}
		& & \int_A f_j D_{R,B}\\
		&=& \int_{A_{SW}} D_{R,B}- \int_{A_{NW}} D_{R,B} + \int_{A_{NE}} D_{R,B} -\int_{A_{SE}} D_{R,B}  \\
		& = &\int_{A_{SW}} [D_{R,B}(x,y)+D_{R,B}(x+\mathsf{w},y+\mathsf{h})-D_{R,B}(x,y+\mathsf{h})-D_{R,B}(x+\mathsf{w},y)]dx dy.
	\end{eqnarray*}
	
	If we apply the definition of $D_{R,B}$ (Equation~\eqref{D_m}) to the four terms inside the integral we get
	
	\begin{eqnarray*}
		\int_A f_j D_{R,B} &=& 
		\int_{A_{SW}}((b-r)[xy+(x+\mathsf{w})(y+\mathsf{h})-x(y+\mathsf{h})-(x+\mathsf{w})y])dx dy\\ 
		&-&
		\int_{A_{SW}}(B[x,y]+B[x+\mathsf{w},y+\mathsf{h}]-B[x,y+\mathsf{h}]-B[x+\mathsf{w},y])dx dy \\
		&+& \int_{A_{SW}}(R[x,y]+
		R[x+\mathsf{w},
		y+\mathsf{h}]-
		R[x,y+\mathsf{h}]-R[x+\mathsf{w},y])dx dy.\\ 
	\end{eqnarray*}
	Observe that we are integrating twice positively and twice negatively over almost identical functions. In fact, the terms of the first integral all cancel out except along the rectangle $[x,x+\mathsf{w})\times [y,y+\mathsf{h})$. Similarly, when we look at the second and third terms, the contribution of any point in $R\cup B$ is cancelled out unless it is in the rectangle $[x,x+\mathsf{w})\times [y,y+\mathsf{h})$. However, by definition of $A$ there are no such points. Thus, we obtain
	
	$$\int_A f_j D_{R,B} = \int_{A_{SW}} (b-r)\mathsf{w}\cdot \mathsf{h} \,dxdy= \int_{A_{SW}} (b-r)2^{-m-2} dxdy = (b-r)2^{-2m-4}$$
	
	That is, when we integrate $f_j D_{R,B}$ over a rectangle $A$ containing no point of $R\cup B$, the result is  $(b-r)2^{-2m-4}$. We know that there are at least $b+r$ rectangles not containing points of $R\cup B$, thus their contribution is at least $\frac{(b+r)(b-r)}{2^{2m+4}} = \frac{(b+r)}{2^{m}}\frac{(b-r)}{16\cdot 2^{m}} \geq \frac{1}{4} \frac{(b-r)}{16\cdot 4 (b+r)} = \Omega(\frac{b-r}{b+r})$. 
\end{proof}

\begin{lemma}\label{lem_upper}
	There exists a constant $c_{2}$ such that $\sum_{k=2}^m \int D_{R,B}G_k \leq c^2 c_{2} \frac{b-r}{b+r} \log (b+r)$. 
\end{lemma}
\begin{proof}
	Recall that $G_k=c^k\sum_{0 \leq j_1 < j_2 < \ldots < j_k \leq m} f_{j_1}\ldots f_{j_k}$. Fix any valid set of indices and consider the value of $\int f_{j_1}\ldots f_{j_k} D_{R,B}$. 
	
	As shown in~\cite{matousek}, function $f_{j_1}\ldots f_{j_k}$ is largely defined by $f_{j_1}$ and $f_{j_k}$. Indeed, if we overlay the rectangular partition defined by functions $f_{j_1},\ldots, f_{j_k}$ we obtain a grid of rectangles whose width is $2^{-j_k}$ and height $2^{-(m-j_1)}$. In each of these rectangles, the function is zero (if any of the rectangles associated to the $f_{j_i}$ functions contains a point of $R\cup B$), or is further subdivided into four equal sized quadrants and in each one it is $+1$ or $-1$ alternatively.  
	
	Let $A$ be one of the rectangles of the refined grid. As shown in Lemma~\ref{lem_lower}, we have that 
	
	$$\int_A f_{j_1}\ldots f_{j_k} D_{R,B} = \tau (b-r) 2^{-2(m+j_k-j_1)-4},$$
	where $\tau \in \{-1,1\}$. This extra term appears because the product of the different functions involved can change the sign of each of the four quadrants. In any case, we have $\int_A f_{j_1}\ldots f_{j_k} D_{R,B} \leq (b-r) 2^{-2(m+g)-4}$ where $g=j_k-j_1$. 
	
	By the way the grid is constructed, there are $2^{m-j_1}\times 2^{j_k}=2^{m+g}$ many rectangles, and thus we conclude that $\int f_{j_1}\ldots f_{j_k} D_{R,B} \leq (b-r) 2^{-m-g-4}$. In order to obtain a bound $\int D_{R,B}G_k$ we sum over all possible indices.
	
	\begin{eqnarray*}
		\int D_{R,B}G_k = c^k \sum_{0 \leq j_1 < j_2 \leq \ldots < j_k \leq m}\int f_{j_1}\ldots f_{j_k}D_{R,B} \leq \frac{c^k (b-r)}{2^{m+4}} \sum_{0 \leq j_1 < j_2 < \ldots < j_k \leq m} 2^{-(j_k-j_1)}.
	\end{eqnarray*}
	
	Note that in the sum, the indices $j_2, \ldots j_{k-1}$ do not matter. Thus, we group the terms by the gap between the indices $j_1$ and $j_k$ (say, if $j_1=3$ and $j_k=7$ the gap is $4$). Note that the minimum gap is at least $k-1$ (since otherwise we do not have enough space to choose the $k-2$ indices in between) and at most $m$. Once we have a gap of $g$ there are $m-g$ options for index $j_1$. 
	
	\begin{eqnarray*}
		\int D_{R,B}G_k &\leq& \frac{c^k(b-r)}{2^{m+4}}\sum_{g=k-1}^m \sum_{j_1=0}^{m-g} \sum_{j_1 < j_2 < \ldots < j_{k-1} < j_1+g}  2^{-g} \\
		&=& \frac{c^k(b-r)}{2^{m+4}}\sum_{g=k-1}^m \sum_{j_1=0}^{m-g} {g-1 \choose k-2}2^{-g} \leq 
		\frac{c^k(b-r) m}{2^{m+4}}\sum_{g=k-1}^m {g-1 \choose k-2}2^{-g}.
	\end{eqnarray*}
	In order to upper bound the sum over all $G_k$, we first reorder the summation order.
	\begin{eqnarray*}
		\sum_{k=2}^m \int D_{R,B}G_k &\leq& \sum_{k=2}^m \frac{c^k(b-r) m}{2^{m+4}}\sum_{g=k-1}^m {g-1 \choose k-2}2^{-g} \\
		& = & \frac{(b-r) m}{2^{m+4}} \sum_{g=1}^m 2^{-g} \sum_{k=2}^{g+1} c^k {g-1 \choose k-2} \\
		& = & \frac{(b-r) m}{2^{m+4}} \sum_{g=1}^m 2^{-g} c^2 \sum_{i=0}^{g-1} {g-1 \choose i} c^i \\
		& = & \frac{(b-r) m}{2^{m+4}} \sum_{g=1}^m 2^{-g} c^2 (1+c)^{g-1} \\
		& =& \frac{(b-r) m c^2}{2^{m+5}} \sum_{g=1}^m \left(\frac{1+c}{2}\right)^{g-1}.\\
	\end{eqnarray*}
	
	The sum contains the first terms of the geometric sum $\sum_{g=1}^\infty \left(\frac{1+c}{2}\right)^{g-1} \leq \frac{2}{1-c}$ (for any $c<1$). In particular, if we set $c \leq 1/2$ we can upper bound the partial sum by $4$. Recall that $m = \Theta(\log(b+r))$ and $2^m = \Theta(b+r)$. Thus, the lemma is proven.
\end{proof}

\begin{cor}
	There exists a constant $\kappa>0$ such that $\int D_{R,B} G 
	\geq \kappa \left(\frac{(b-r)\cdot \log(b+r)}{b+r}\right)
	$.
\end{cor}
\begin{proof}
	Apply the inequality $\int (A+B) \geq \int A -\int |B|$ and Lemmas~\ref{lem_lower} and \ref{lem_upper} to obtain:
	
	$$ \int D_{R,B} G = \int D_{R,B}G_1 + \sum_{k=2}^m \int D_{R,B}G_k  \geq c(c_1-cc_{2})\left(\frac{(b-r)\cdot \log(b+r)}{b+r}\right)$$
	
	Note that Lemmas~\ref{lem_lower} and \ref{lem_upper} holds for any value of $c$ such that $c\in (0, 1/2]$. By choosing a sufficiently small value of $c$ (say, $c=\min\{\frac{1}{2}, \frac{c_1}{2c_{2}}\}$) we obtain
	
	$$ \int D_{R,B} G \geq \frac{cc_1}{2}\left(\frac{(b-r)\cdot \log(b+r)}{b+r}\right)$$
	
\end{proof}
This completes the proof of \cref{theo_discrep}.

When $R = \emptyset$, it would be expected that we need to distribute the blue points uniformly in the unit square to have a low discrepancy. Indeed, it is also held for the red points. The following theorem implies that even if there are many red points, but the red points are concentrated in the lower half of the unit square, the discrepancy cannot be reduced. For simplicity, we only show a special case of how the discrepancy is depended on the points in $[0,1] \times [1/2, 1]$, which is good enough for our purpose in \cref{3D-lower-bound}. Notice that the same argument can be applied in a more general case.

\begin{theorem}\label{theo_discrep_2}
	For any set $R$ and $B$ of points in the unit square such that $|R|=r$, $|B|=b$ and $b>r$. Let $r_2$ and $b_2$ be the number of red and blue points in $[0,1] \times [1/2, 1]$ respectively. 
	It holds that 
	\begin{equation*} 
		D_{R,B}^* = \Omega\left(\frac{(b_2-r_2)\cdot \log(b_2+r_2)}{b_2+r_2}\right).
	\end{equation*}
\end{theorem} 

\begin{proof}
	Let $R_2$ and $B_2$ be the set of red and blue points in $[0,1] \times [1/2, 1]$ respectively.
	Consider the upper half of the unit square $[0,1] \times [1/2, 1]$ and rescale the vertical length to be $1$. By \cref{theo_discrep}, there exists a point $(x,2y)$ such that $| D_{R_2,B_2}(x,2y) | = |2xy(b_2-r_2) - (B_2[x,2y] - R_2[x,2y]) | \geq 2c\left(\frac{(b_2-r_2)\cdot \log(b_2+r_2)}{b_2+r_2}\right)$ for some constant $c$.		
	
	Then, we map the point $(x,2y)$ back to a point $(x,1/2+y)$ in the original unit square. We will show that either $D_{R,B}(x,1/2+y)$ or $D_{R,B}(x,1/2-\epsilon)$ would give us the desired lower bound, where $\epsilon$ is an arbitrarily small constant such that rectangle $[0,1] \times [0, 1/2-\epsilon]$ only contains $B \setminus B_2$ and $R \setminus R_2$.
	
	If $|D_{R,B}(x,1/2+y)| \geq  c\left(\frac{(b_2-r_2)\cdot \log(b_2+r_2)}{b_2+r_2}\right)$, we are done.
	
	If $|(b-r)/2 - (b_2-r_2) | \geq c/4 \left(\frac{(b_2-r_2)\cdot \log(b_2+r_2)}{b_2+r_2}\right)$, the proof is also done. Because 
	\begin{eqnarray*}
		& &	|D_{R,B}(1, 1/2-\epsilon)|\\
		& \stackrel{(\ref{D_m})}{=} & |(b-r)(1/2-\epsilon) - (B[1,1/2-\epsilon] - R[1, 1/2-\epsilon])|\\
		& = & |(b-r)(1/2-\epsilon) - (b-r - (b_2 - r_2))|\\
		& = & |(b-r)/2 - (b_2-r_2) - (b-r)\epsilon|\\
		& > & c/8 \left(\frac{(b_2-r_2)\cdot \log(b_2+r_2)}{b_2+r_2}\right).
	\end{eqnarray*}	
	
	Suppose that the above two cases do not hold, we have $|D_{R,B}(x,1/2+y)| < c\left(\frac{(b_2-r_2)\cdot \log(b_2+r_2)}{b_2+r_2}\right)$ and $|(b-r)/2 -(b_2-r_2 ) | < c/4 \left(\frac{(b_2-r_2)\cdot \log(b_2+r_2)}{b_2+r_2}\right)$.
	Let $R_1 = R \setminus R_2$ and $B_1 = B \setminus B_2$, which are inside the rectangle $[0,1] \times [0, 1/2-\epsilon]$. Consider 
	\begin{eqnarray*}
		& & D_{R,B}(x,1/2+y)\\
		& = & (b-r)x(1/2+y) - (B[x,1/2+y] - R[x,1/2+y])\\
		& = & (b-r)x(1/2+y) - (B_2[x,1/2+y] - R_2[x,1/2+y] + B_1[x,1/2-\epsilon] - R_1[x,1/2-\epsilon])\\
		& = & (b - r)x(1/2-\epsilon) - (B_1[x,1/2-\epsilon] - R_1[x,1/2-\epsilon]) + (b-r)x\epsilon\\
		& & + (b-r)xy -(B_2[x,1/2+y] - R_2[x,1/2+y] )\\
		& = & D_{R,B}(x,1/2-\epsilon) + (b-r)xy -(B_2[x,1/2+y] - R_2[x,1/2+y] ) + (b-r)x\epsilon\\
		& > & D_{R,B}(x,1/2-\epsilon) + 2(b_2-r_2)xy - (B_2[x,1/2+y] - R_2[x,1/2+y] )\\
		&   & - c/2 \left(\frac{(b_2-r_2)\cdot \log(b_2+r_2)}{b_2+r_2}\right) + (b-r)x\epsilon \\
		& = & D_{R,B}(x,1/2-\epsilon) +  D_{R_2,B_2}(x,2y) - c/2 \left(\frac{(b_2-r_2)\cdot \log(b_2+r_2)}{b_2+r_2}\right) + (b-r)x\epsilon
	\end{eqnarray*}	
	The first inequality is given by $b-r > 2(b_2-r_2) - c/2 \left(\frac{(b_2-r_2)\cdot \log(b_2+r_2)}{b_2+r_2}\right)$.
	Since $|D_{R,B}(x,1/2+y)| < c\left(\frac{(b_2-r_2)\cdot \log(b_2+r_2)}{b_2+r_2}\right)$ and $| D_{R_2,B_2}(x,2y) | \geq 2c\left(\frac{(b_2-r_2)\cdot \log(b_2+r_2)}{b_2+r_2}\right)$, we can conclude that $|D_{R,B}(x,1/2-\epsilon)| = \Omega(\frac{(b_2-r_2)\cdot \log(b_2+r_2)}{b_2+r_2})$.
\end{proof}



\section{Lower bound for two dimensional weak CDRs}\label{TradeoffHausdorffErrorDeathVertices}

Before giving the proof of \cref{theo_lower_2}, we recall that a proof for a proper CDR (i.e., one without inner leaves) was given in~\cite{cknt-cdg-09j}. Our proof follows the same spirit, so we first give an overview of their proof and describe what changes when we introduce inner leaves. 

\begin{lemma}\label{lem_distance}
	Given a CDR, a point $p=(x,y)\in L_N$, and an integer $n<N$, let $p'=(x',y')\in L_n$ be the unique point of $L_n$ that is in $dig(o,p)$. The Hausdorff error of the CDR is at least $|x'-x\cdot \frac{n}{N}|$. 
\end{lemma}

\begin{proof}
	This result was shown by Chun {\em et al.}~\cite{cknt-cdg-09j} (Lemma 3.5, in Cases 1 and 2). We give the proof for completeness. Consider the $L$-infinity ball of radius $|x'-x\cdot\frac{n}{N}|$ centered at $p\cdot\frac{n}{N}$. By construction, this ball contains $p'$ in its boundary. Because of the monotonicity axiom, no vertex of $dig(o,p)$ can be in the interior of the ball. 
	In particular, when measuring the Hausdorff distance of point $p \cdot \frac{n}{N}\in \overline{op}$ we get an error of at least $|x'-x \cdot\frac{n}{N}|$.
\end{proof}

Consider any point $p\in L_N$ and virtually sweep a line of slope $-1$ from the origin all the way to $L_N$. During the sweep, the intersection between the diagonal line and either the Euclidean segment $\overline{op}$ or the digital one $dig(o,p)$ will be a point. Lemma~\ref{lem_distance} says that if we can find an instant of time for which two intersection points are at distance $\partial$ from each other, then the Hausdorff error of the whole CDR must be $\Omega(\partial)$ (see \cref{fig:snalpha}).

In order to find this instant of time we see how much the subtrees grow. Consider a consecutive set of $I$ vertices in some intermediate layer $L_n$. Let $\mathcal{L}(I)$ be the vertices of $L_N$ whose digital path to the origin passes through some vertex of $I$. If the CDR has small error, we need $\mathcal{L}(I)$ to have roughly $\frac{N}{n}|I|$ many points. The difference between the expected number of vertices and $|\mathcal{L}(I)|$ combined with Lemma~\ref{lem_distance} will give a lower bound on the Hausdorff error.

\begin{figure}
	\centering
	\includegraphics[scale=0.4]{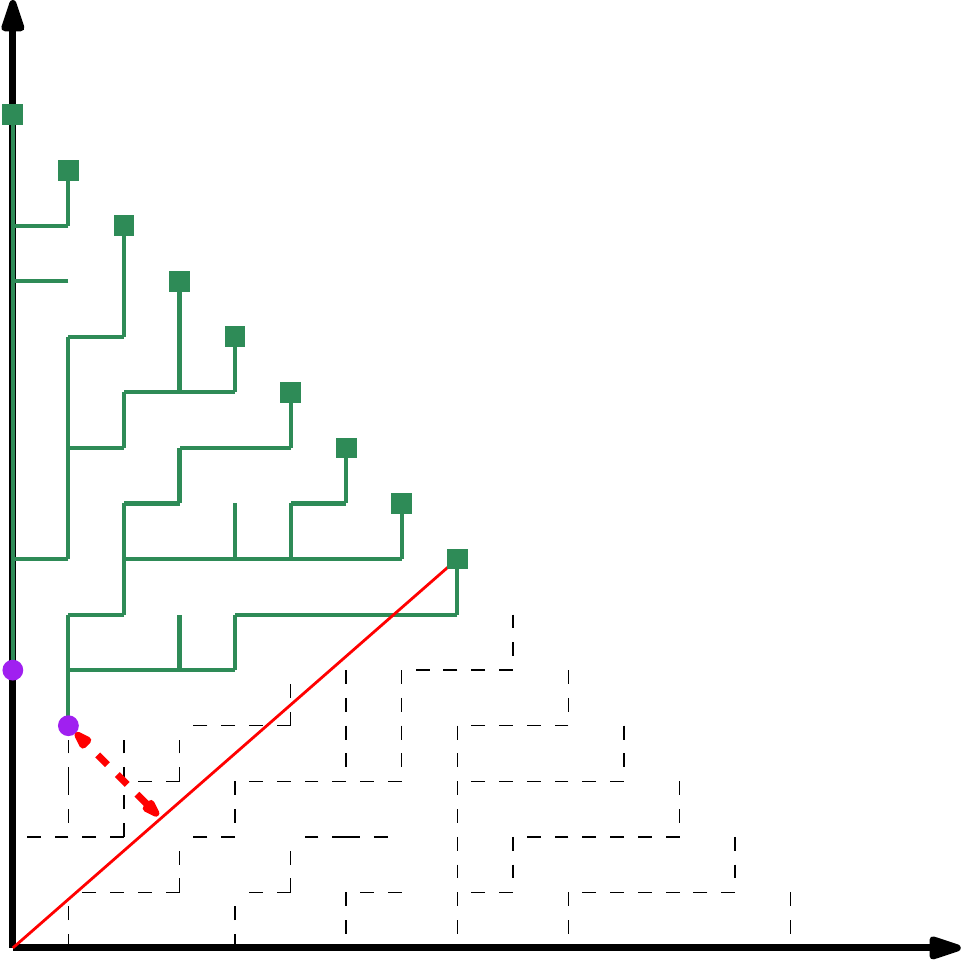}
	\caption{Illustration of why the two sets $I$ (purple disks) and $\mathcal{L}(I)$ (green squares) should have proportional sizes. If the size of $\mathcal{L}(I)$ grows drastically (as shown in the figure), the point of the highest $x$-coordinate in $\mathcal{L}(I)$ must make a significant detour to pass through $I$, causing a large error. A similar effect happens if the size of $\mathcal{L}(I)$ is comparatively small.}
	\label{fig:snalpha}
\end{figure}

Our proof follows the same spirit (transform the tree into a pointset, use discrepancy to find a subset with too many/too few children and use Lemma~\ref{lem_distance} to find a large error). Although all three steps follow the same spirit, they need major changes to account for the possibility of inner leaves.

The biggest change is how we map the tree. In proper CDRs each line has a unique split vertex and always extends to $L_N$. Thus, a region with a large number of split vertices directly implies a large error. In our setting, we could potentially have a region with many split vertices followed by a large number of inner leaves to cancel out the growth. This is why we need two major changes: first we now color the points red and blue depending on whether they are split vertices or inner leaves. We also introduce a second dimension to track when the children of a split vertex stop extending.  
Intuitively speaking, the $x$-coordinate of our mapping will be similar to the mapping done by Chun {\em et al.}~\cite{cknt-cdg-09j} whereas the $y$-coordinate represents time. Thus, the difference in $y$-coordinates between red and blue points can be used to determine for how long are the two children of a split vertex alive (the longer the difference in $y$-coordinates, the further away that the two children extend).

We now use the mapping of \cref{TransformationSection} together with the two colors 
discrepancy (\cref{theo_discrep,theo_discrep_2}) to show a lower bound on the error of weak CDRs. 
The discrepancy result in \cref{theo_discrep} considers the points in the whole unit square. Due to some technical reasons, in \cref{3D-lower-bound} we will need a discrepancy result for the points in the {\em upper half} of the unit square instead (\cref{theo_discrep_2}). The difference between the two theorems is just a constant factor and thus would have little implication. Here we use \cref{theo_discrep_2} and prove the result in terms of the number of inner leaves in the upper half. Specifically, we show the following result.
\weakCDR*

\begin{proof}
	Given a weak CDR and its associated tree $T$, consider its transformation into the sets $R$ and $B$ of red and blue points defined by $\pi$. 
	Let $b_2$ and $r_2$ be the numbers of blue and red points in the rectangle $[0,1] \times [1/2,1]$ respectively. By Lemma~\ref{lem_one_line}, we have $b_2 - r_2 = \lfloor N/2 \rfloor$.
	We apply the discrepancy result (\cref{theo_discrep_2}) with $b_2-r_2 = \lfloor N/2 \rfloor$ and $r_2 = \kappa_2$, and obtain that there exists $\alpha, \beta \in [0,1]$ such that $\abs{B[\alpha,\beta] - R[\alpha,\beta] - N \cdot \alpha \cdot \beta} > c' \cdot \frac{N \cdot \log N}{N+\kappa_2}$. 
	
	We want to use \cref{theo_interp}
	on the vertex of $T$ whose image is $(\alpha,\beta)$.
	Naturally, such a vertex need not exist, but we will find one nearby whose associated discrepancy is also high. Let $n= \lfloor N \cdot \beta \rfloor$ and observe that $B[\alpha,\beta]=B[\alpha,\frac{n}{N}]$; indeed, by the way we transform points, their $y$-coordinates are of the form $i/N$. However, by definition of $n$ we know that $\beta$ is between $n/N$ and $(n+1)/N$ and thus no point can lie in the horizontal strip $y\in (n/N, \beta]$ (by the same argument we also have $R[\alpha,\beta]=R[\alpha,\frac{n}{N}]$). 
	
	If we substitute $\beta$ in the previous equation we get
	\begin{equation*}
		\abs{B\left[\alpha,\frac{n}{N}\right] - R\left[\alpha,\frac{n}{N}\right]-\alpha n} > c' \cdot \frac{N \log N}{N+\kappa_2} -1 \geq c'' \cdot \frac{N \cdot \log N}{N+\kappa_2} 
	\end{equation*}
	for a large enough $N$, $\kappa_2 \in O(N \log N)$ and for some $c''>0$. We get the additional $1$ term because of the rounding in the definition of $n$.
	
	Now we need to do a similar operation for $\alpha$. Let $q_i=(i,n-i)$ be a vertex of $L_n$. By Lemma~\ref{lem_mapsorted} the image of the auxiliary function $M(q_i)$ monotonically increases as $i$ grows. Let $Q=\{q_i\colon M(q_i) \leq \alpha\}$ and $\alpha'=max_{q_i\in Q} M(q_i)$. Note that, by definition of the set $Q$, it trivially holds that $\alpha'\leq \alpha$.
	
	\begin{lemma}
		\label{lem:approx_alpha}
		$B[\alpha,\frac{n}{N}] - R[\alpha,\frac{n}{N}]=B[\alpha',\frac{n}{N}] - R[\alpha',\frac{n}{N}]$
	\end{lemma}
	
	\begin{proof}
		The difference between the two rectangles is the rectangle $\Delta$ whose opposite corners are $(\alpha',0)$ and $(\alpha, n/N)$, and one of the boundary $\overline{(\alpha',0)(\alpha',\frac{n}{N})}$ is open. 
		We claim that red and blue points are paired (sharing the same $x$-coordinate) in $\Delta$ (and thus, for each red point that we remove we are also removing a blue one). 
		By Lemma~\ref{lem_bijection}, we know that all the blue points have different $x$-coordinates, so do red points. Hence, if there are red and blue points on the same vertical line, they must be the only pair in that vertical line. First notice that if there is a red point in $\Delta$, there also exists a blue point in $\Delta$ with the same $x$-coordinate and below the red point.
		By the virtual walk that we define the auxiliary function, every split vertex is closer to the origin than the corresponding leaf. Hence, after the transformation $\pi$, if there is a red point, then there must exist a blue point with the same $x$-coordinate (by Lemma~\ref{lem_bijection}) and smaller $y$-coordinate. 
		Then, we will show that if there is a blue point in $\Delta$, there also exists a red point in $\Delta$ with the same $x$-coordinate.
		
		Assume, for the sake of contradiction that there exists a blue point $p$ in $\Delta$ such that there does not exist a red point $q$ with the same $x$-coordinate as $p$ in $\Delta$. 
		Let $s$ be the split vertex whose image is $p$. By definition of the transformation $\pi$, the $x$-coordinate of $p$ is $M(s)$, which is between $\alpha'$ and $\alpha$. We apply Lemma~\ref{lem_bijection} to find the unique leaf $\ell$ such that $M(s)=M(\ell)$. 
		Since $\pi(\ell)\not\in\Delta$, we have that $\ell_x+\ell_y>n$. Let $m$ be the unique vertex of $L_n$ that is in the path from $s$ to $\ell$. It follows that $\pi(m)=(M(\ell), \frac{n}{N})\in \Delta$. This gives a contradiction with the definition of $\alpha'$, and thus implies that if there exists a blue point in $\Delta$, then there also exists a red point in $\Delta$ with the same $x$-coordinate.
	\end{proof}
	
	Thus, given a pair $(\alpha,\beta)$ whose associated rectangle has high discrepancy, we have {\em snapped} it to the pair $(\alpha', \frac{n}{N})$ that defines another rectangle with high discrepancy. More importantly, by definition of $Q$, we know that $\pi(q_{|Q|-1})=(\alpha', \frac{n}{N})$. Note that $q_{|Q|-1}$ need not be a split vertex or an inner leaf (and thus, $(\alpha', \frac{n}{N})$ may not be a point of $R \cup B$).
	
	Let $b'=B[\alpha',\frac{n}{N}]$ and $r'=R[\alpha',\frac{n}{N}]$. 
	If we apply \cref{theo_interp} to point $q_{|Q|-1}$ we get that $b' - r' - 2 \leq |Q|-1 \leq b'-r'$. This set $Q$ is the one that makes the role of $I$ in the proof overview: we know that vertices of $Q$ are the ones that extend to cover all the vertices of $L_N$ whose image is $\alpha'$ or less. As such, we would expect $|Q|$ to contain roughly $n\alpha'$ elements. However, the discrepancy result tells us that the size of $Q$ is $c''\frac{N\log N}{N+\kappa_2}$ units away from that value.
	We say that $p$ is {\em productive} if some point of $T(p)$ is in $L_N$ (this is equivalent to the fact that $p$ can be extended to reach the boundary).
	Let $k \leq b'-r'-2$ be the biggest integer such that $q_k$ is productive. Note that $k$ is well defined because $q_0$ is always productive ($(0,n)$ always extends to $(0,N)$).
	The proof now considers a few cases depending on whether $k$ is small or large (specifically, we say that $k$ is {\em small} if $|Q|-1-k \geq (b'-r'-2)-k> \frac{c''}{2} \cdot \frac{N  \log N}{N+\kappa_2}$, {\em large} otherwise) and if $Q$ contains too few or too many points.
	
	\begin{figure}
		\centering
		\includegraphics[width=0.7\textwidth]{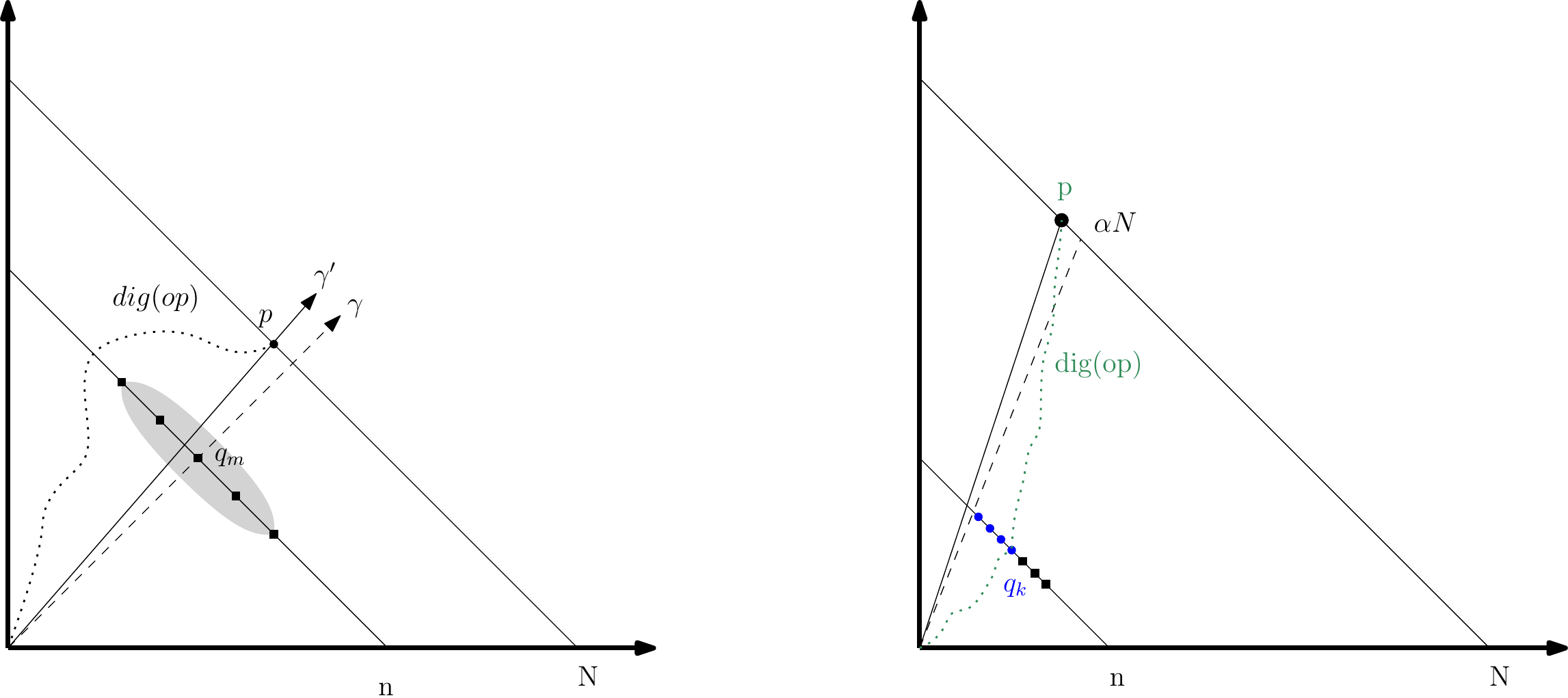}
		\caption{(left) When $k$ is small we have  $\Omega(\frac{N\log N}{N+\kappa_2})$ consecutive vertices in $L_n$ that are not productive (shown as squares). In particular, the ray $\gamma$ through the middle point must make a large detour. (right) When $k$ is large, there is a digital path through $q_k$ with a big detour.}
		\label{fig:Ksmall}
	\end{figure}
	
	\begin{description}
		\item[$k$ is small] Recall that we looked for the largest possible $k$ (such that $q_k$ is  productive). Thus, if $k$ is small, we have many points in layer $L_n$ that are consecutive and not productive. In particular, none of the vertices in $q_{b'-r'-\lfloor\frac{c''}{2} \cdot \frac{N \log N}{N+\kappa_2}\rfloor}, \ldots, q_{b'-r'-2}$ are productive. Let $q_m=q_{b'-r'-\lfloor\frac{c''}{4} \cdot \frac{N  \log N}{N+\kappa_2}\rfloor}$ (note that this point is surrounded by non-productive points in both sides along $L_n$).
		
		Shoot a ray $\gamma$ from $o$ towards $q_m$.
		Let $p$ be the vertex on $L_N$ that is closest to $\gamma$. Observe that the $|| \cdot ||_{\infty}$ distance between $\gamma$ and $p$ is at most $1/2$. Let $\gamma'$ be the ray shooting from $o$ towards $p$. Similarly, the $|| \cdot ||_{\infty}$ distance between $\gamma'$ and $q_m$ is at most $1/2$ (see \cref{fig:Ksmall}, left). 
		
		
		We now apply Lemma~\ref{lem_distance} to $dig(o,p)$. We know that the Euclidean segment $\overline{op}$ is close to $q_m$. The digital segment must cross $L_n$ and is far from $q_m$ (the closest it can pass is either $q_{b'-r'-\lfloor\frac{c''}{2} \cdot \frac{N  \log N}{N+\kappa_2}\rfloor-1}$ or $q_{b'-r'-1}$). That is, we know that the intersection of $\overline{op}$ with the line $x+y=n$ is at most half a unit away from $q_m$. Similarly, the intersection with $dig(o,p)$ is at least $\lfloor\frac{c''}{4} \cdot \frac{N  \log N}{N+\kappa_2}\rfloor$ from $q_m$. Thus, by triangle inequality the $|| \cdot ||_{\infty}$ distance between $dig(o,p)$ and $\overline{op}$ is at least $\lfloor\frac{c''}{4} \cdot \frac{N \log N}{N+\kappa_2}\rfloor-3/2 \in \Omega (\frac{N\log N}{N+\kappa_2})$. 

		\item[$k$ is large and $b'-r'\geq n\alpha+c''\cdot \frac{N \log N}{N+\kappa_2}$] Look at the $x$-coordinate of $q_k$. We know that $Q$ has at least $b'-r'-1\geq n\alpha+c''\cdot \frac{N \log N}{N+\kappa_2}-1$ many elements, and $k$ is among the productive vertices with the largest $x$-coordinate. In particular, the $x$-coordinate of $q_k$ is at least $b'-r' -2 \geq n\alpha+ \frac{c''}{2}\cdot \frac{N  \log N}{N+\kappa_2}-2$.
		
		Let $p$ be the unique leaf of $L_N$ such that $M(p)=M(q_k)$. We now apply Lemma~\ref{lem_distance} to $dig(o,p)$ at the line $x+y=n$. By definition of $p$, we have that $dig(o,p)$ passes through $q_k$. Now, by definition of $Q$, we know that $M(q_k) \leq \alpha$ and in particular the $x$-coordinate of $p$ is at most $\alpha N$ (see \cref{fig:Ksmall}, right). Thus, the Euclidean segment $\overline{op}$ must intersect at a point whose $x$-coordinate is at most $\alpha n$. 
		
		That is, when we look at the Euclidean and the digital segments along line $x+y=n$, the Euclidean crossing happens at $x$-coordinate at most $\alpha n$. However, the $x$-coordinate of the digital crossing is at least $\alpha n + \frac{c''}{2}\cdot \frac{N  \log N}{N+\kappa_2} -1$. By Lemma~\ref{lem_distance} we conclude that the error must be $\Omega(\frac{N  \log N}{N+\kappa_2})$ as claimed.
		
		\item[$b'-r'< n\alpha - c''\cdot \frac{N \log N}{N+\kappa_2}$] This proof is very similar to the previous case. Consider the vertex $p=(\lfloor \alpha N \rfloor, N-\lfloor \alpha N \rfloor)\in L_N$ and apply Lemma~\ref{lem_distance} to $dig(o,p)$ and $\overline{op}$. 
		
		At line $x+y=n$ the Euclidean segment $\overline{op}$ passes through a point whose $x$-coordinate is $\lfloor \alpha N \rfloor \cdot \frac{n}{N} \geq \lfloor \alpha n \rfloor -1$. By definition, $M(p)\leq \alpha$ and thus $dig(o,p)$ must pass through some vertex $q$ of $Q$. In particular, the $x$-coordinate of $q$  is at most $b'-r' < n\alpha - c''\cdot \frac{N  \log N}{N+\kappa_2}$, giving the $\Omega(\frac{N  \log N}{N+\kappa_2})$ error and completing the proof of \cref{theo_lower_2}.\qedhere
	\end{description}
\end{proof}

Note that if we use \cref{theo_discrep} instead, the same argument follows and we would get the following result.
\begin{theorem}\label{theo_lower}
	For any $N \in \N$, any weak CDR defined on $\domain \subset \Z^2$ with $\kappa_1$ inner leaves has $\Omega(\frac{N  \log N}{N+\kappa_1})$ error.
\end{theorem}


\section{Lower bound for CDRs in high dimensions}
\label{3D-lower-bound}

We now use the lower bound of weak CDRs to obtain a lower bound for CDRs in three or higher dimensions. Consider the restriction of any $d$-dimensional CDR $T$ to the $x_1x_2$-plane (we call this restriction the {\em $x_1x_2$-restriction of T} and denote it by $T_{x_1x_2}$). Recall that the key observation is that $T_{x_1x_2}$ is a (possibly weak) CDR and that any inner leaf in $T_{x_1x_2}$ must extend in some $x_i$-direction in $T$ for some $i \in [3..d]$. We have seen that $T_{x_1x_2}$ needs to have a large number of inner leaves to have $o(\log N)$ error. In the following, we will show that a large number of inner leaves will cause constraints for $\Z^d$ and have an impact in the overall error of $T$.

We do a slight abuse of notation and use the same terms as in two dimensions. For simplicity of the notation, we assume that $N$ is a positive even number. 
For any $n \leq N$, let $L_n = \{(x_1,x_2,\ldots,x_d) \in \domain \colon \sum_{i=1}^{d} x_i = n\}$. Given any CDR in $\domain$, we consider the CDR as a tree rooted at the origin. Let $T(v)$ be the subtree rooted at $v$.

From Theorem~\ref{theo_lower_2}, we already know that in order for $T_{x_1x_2}$ to have sublogarithmic error we must have $\kappa_2\in \omega(N)$ inner leaves. However, each inner leaf ties to a boundary leaf in $L_N$ in $d$ dimensions. In other words, the subtrees rooted at the vertices in $L_{N/2-1} \cap T_{x_1x_2}$ must cover all these boundary vertices.
We now observe that a weak CDR with inner leaves in the $x_1x_2$-plane induces subtrees which are too big for the high dimensional proper CDR (See Figure~\ref{fig:3derror}). 

\begin{lemma}
	\label{lem:subtree_dist}
	Given any CDR in $\domain$, let $\kappa_2$ be the number of inner leaves in $T_{x_1x_2}$ between $L_{N/2}$ and $L_{N}$. There exists a vertex $v \in L_{N/2-1}$ 
	such that $v_i=0$ for $i=3,\dots,d$ and some boundary leaf $u \in T(v) \cap L_N$ has $u_j \geq (\kappa_2/N)^{\frac{1}{d-2}}-1$ for some $j \in [3..d]$.
\end{lemma}

\begin{proof}	
	\begin{figure}
		\centering
		\includegraphics[width=.6\textwidth]{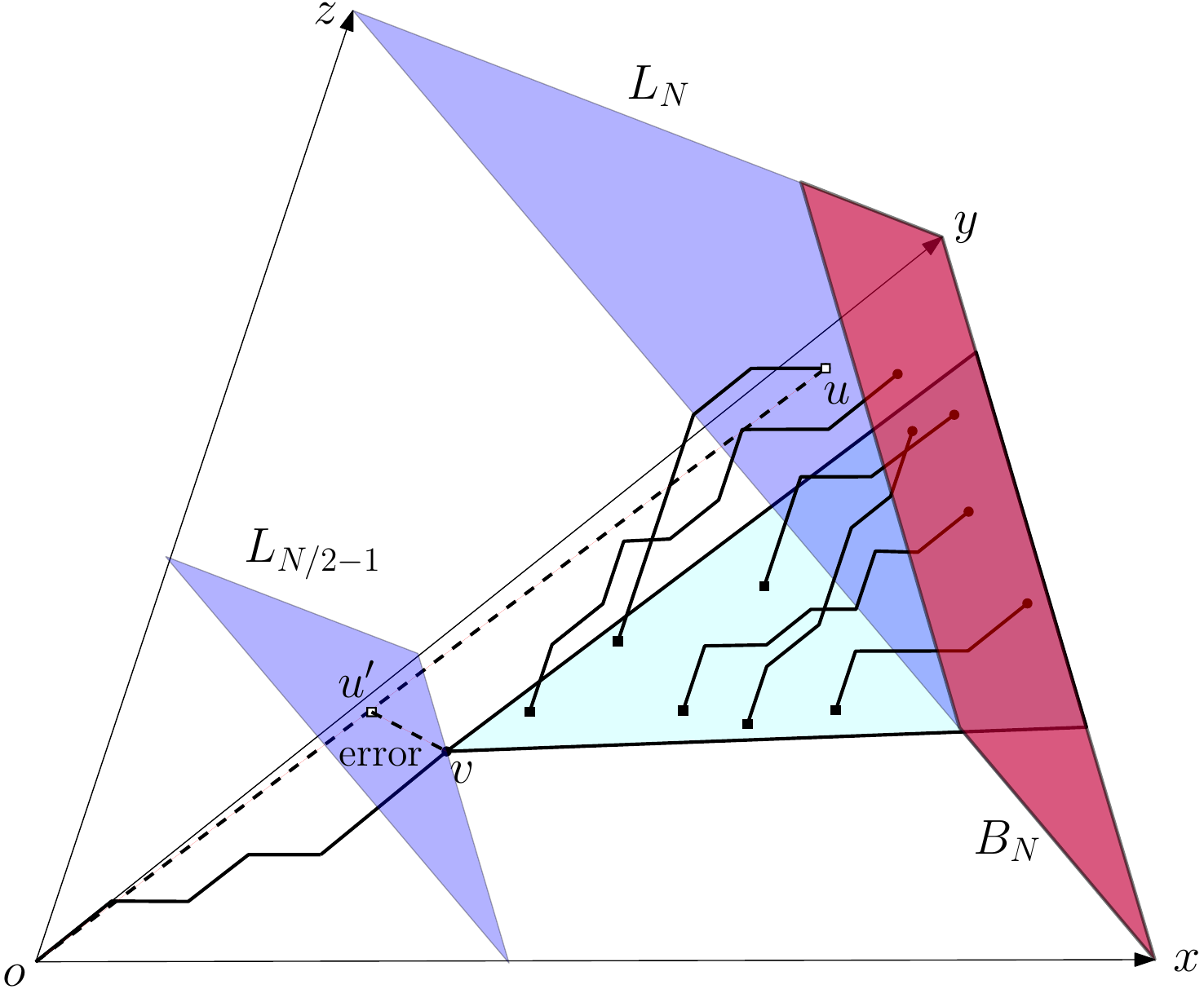}
		\caption{Illustration of Lemmas~\ref{lem:subtree_dist} and \ref{lem:H_dist}: the red region represents the region of $B_N$. If we have lots of inner leaves in $T_{xy}$, it will have many descendants in the three dimensional CDR at layer $L_N$ so that the height of the red region attempting to contain them is large. In particular, we can find a vertex $v$ on the $xy$-plane such that $v$ is on the $dig(o,u)$ and $u$ is far away from the $xy$-plane. For simplicity, we show the Euclidean error between $v$ and $u'$, but we note that the proof argues under the $|| \cdot ||_{\infty}$ metric.}
		\label{fig:3derror}
	\end{figure}
	
	The proof follows from a packing argument. Consider the set $V=\{(0,N/2-1,0,\ldots,0), (1,N/2-2,0,\ldots,0), \ldots, (N/2-1,0,0,\ldots,0)\}$. Note that these vertices lie in the $x_1x_2$-plane and thus are in $T_{x_1x_2}$. Because they are the two dimensional equivalent of $L_{N/2-1}$, the union of their subtrees covers $T_{x_1x_2}$ between $N/2$ and $N$. In this region we know that we have $\kappa_2$ many inner leaves, which will extend to $L_N$ with the first step in the $x_i$-direction for some $i \in [3..d]$. Let $Y_N$ be the extended vertices on $L_N$ from these $\kappa_2$ inner leaves, i.e., $|Y_N| \geq \kappa_2$.
	
	Let $B_N = \{(x_1,x_2,\ldots,x_d)\in \domain \colon \sum_{i=1}^{d} x_i = N, x_1+x_2<N \mbox{ and }\forall i \in [3..d], x_i < (\kappa_2/N)^{\frac{1}{d-2}}-1\}$, see \cref{fig:3derror}. Since we have less than $(\kappa_2/N)^{\frac{1}{d-2}}$ choices for $x_3,\ldots,x_d$, at most $N$ choices for $x_1$ and the value of $x_2$ is adjusted to satisfy the constraint $\sum_{i=1}^{d} x_i = N$, the size of $B_N$ is less than $\kappa_2$. Hence, $B_N$ cannot contain all vertices of $Y_N$. Moreover, no vertices of $Y_N$ lie on $x_1x_2$-plane, so there exists some vertex $u \in Y_N$ such that $u_j \geq (\kappa_2/N)^{\frac{1}{d-2}}-1$ for some $j \in [3..d]$, which is in $T(v) \cap L_N$ for some $v \in V$.
\end{proof}

The existence of this vertex $v$ is the root of the problem.
We conclude with the following statement.

\begin{lemma}
	\label{lem:H_dist}
	Any CDR defined on $\domain \subset \Z^d$ with $\kappa_2$ inner leaves in $T_{x_1x_2}$ between $L_{N/2}$ and $L_{N}$ has  $\Omega((\kappa_2/N)^{\frac{1}{d-2}})$ error.
\end{lemma}

\begin{proof}
	Apply Lemma~\ref{lem:subtree_dist} to obtain a vertex $v\in L_{N/2-1} \cap T_{x_1x_2}$ that satisfies some $u \in T(v) \cap L_N$ with $u_j \geq (\kappa_2/N)^{\frac{1}{d-2}}-1$ for some $j \in [3..d]$. 
	Let $u'$ be the intersection of $\overline{ou}$ and the affine plane containing $L_{N/2-1}$, see \cref{fig:3derror}. As $L_N$ and $L_{N/2-1}$ are parallel, $\frac{u_j'-o_j}{u_j-o_j} \geq \frac{1}{3}$ for $N\geq 6$, this implies that $u'_j = \Omega((\kappa_2/N)^{\frac{1}{d-2}})$. By construction, we have that $v$ is on the $dig(o,u)$ and $v_j = 0$, hence $|| \cdot ||_{\infty}$ distance between $dig(o,u)$ and $\overline{ou}$ is $\Omega((\kappa_2/N)^{\frac{1}{d-2}})$.
\end{proof}

Combining with \cref{theo_lower_2} gives us a lower bound for CDRs in $d$ dimensions.

\lowerCDR*

	\begin{proof}
		By \cref{theo_lower_2} and Lemma~\ref{lem:H_dist}, the error is $\Omega(\frac{N\log N}{N+\kappa_2})$ and $\Omega((\kappa_2/N)^{\frac{1}{d-2}})$, where $\kappa_2$ is the number of inner leaves in $T_{x_1x_2}$ between $L_{N/2}$ and $L_{N}$. The balance between the two is obtained by choosing $\kappa_2 = \Theta(N \log^{\frac{d-2}{d-1}} N)$, giving the  $\Omega( \log^{1/(d-1)} N)$ lower bound.
	\end{proof}



\section{A construction of a weak CDR with constant error} \label{GreedyCDR}

\begin{figure}
	\centering
	\includegraphics[width=\textwidth]{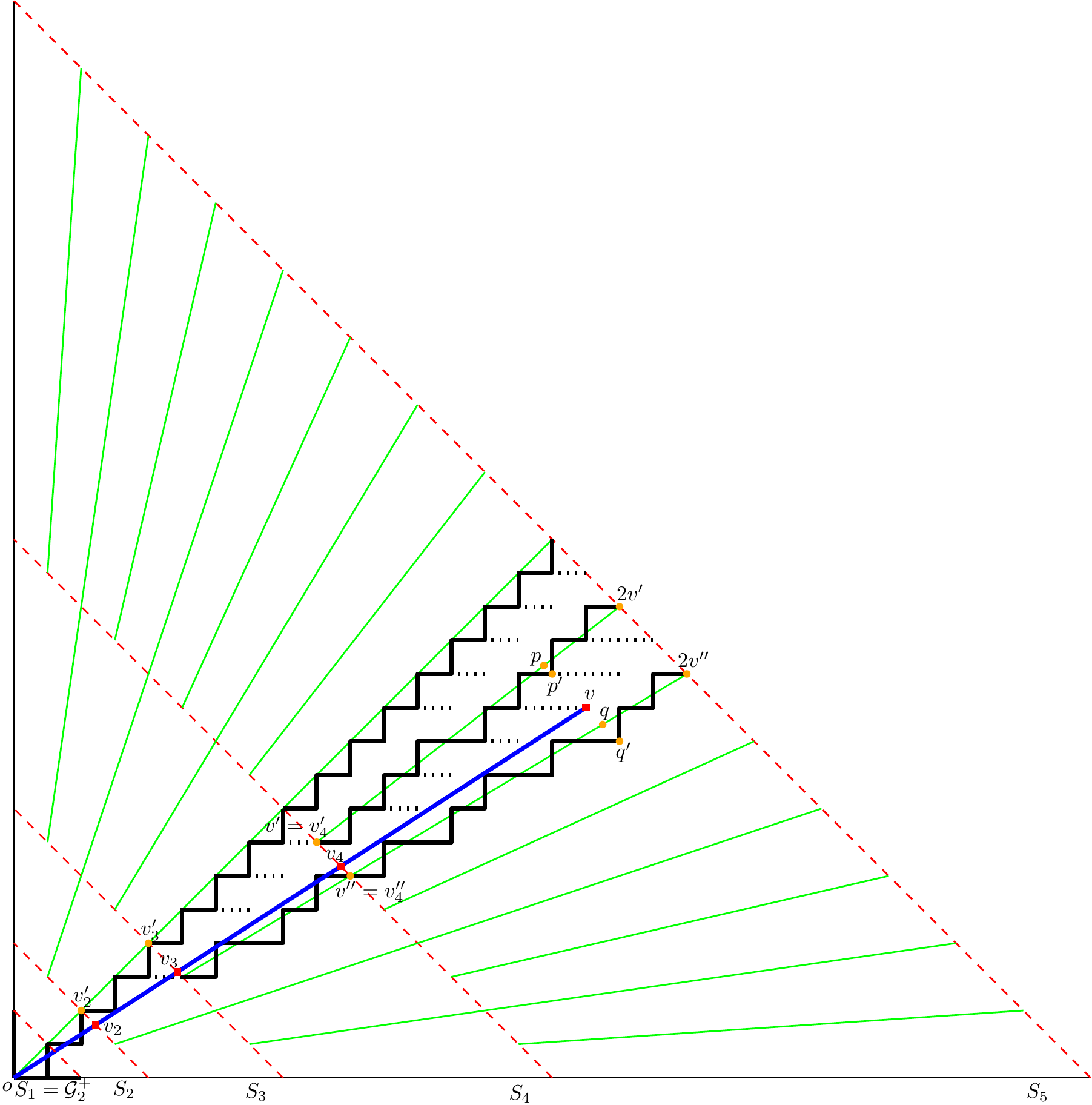}
	\caption{Outline of GREEDY. In the example, $N = 32$. $\domain$ is partitioned into $5$ slices by the red dash lines. The green solid lines indicate where the greedy paths are constructed between $(x,y)$ and $(2x,2y)$. A digital path $dig(o,v)$ is shown through the concatenation of some greedy paths as well as the notation used in the proof of the Hausdorff error for $dig(o,v)$.}
	\label{fig:greedy}
\end{figure}

In this section we describe how to construct a weak CDR in 2 dimensions with the smallest possible error. Specifically, we show that the weak CDR has at most $5/2$ error and at most $N^2/12$ inner leaves, which is about $1/6$ of the total number of grid points in $\domain$.

Assume that $N$ is a power of 2. We partition $\domain$ into $\log_2 N$ diagonal slices by a set of lines $x+y=2^i$ for $i=1,\ldots, \log_2 N-1$. We use $S_i$ to denote the $i$-th slice between $x+y=2^{i-1}$ and $x+y=2^i$ for $i=2,\ldots, \log_2 N$.
The first slice $S_1$ is $\mathcal{G}^+_2$, which only has 6 points. There are two proper CDRs for this small set and both have the same error. So we can use either of the two. For each other slice $S_i$, we draw a {\em greedy} digital path from each point $p=(p_x,p_y) \in L_{2^{i-1}}$ to $2p=(2p_x, 2p_y)$. The greedy digital path simply tries to approximate the Euclidean segment as much as possible. More formally, the path between $p$ and $2p$ is defined by picking a point in each $L_j$ that has the smallest $|| \cdot ||_{\infty}$ distance to the line segment $\overline{p,2p}$ for $j=2^{i-1}, \ldots, 2^{i}$ (in case of tie, we pick the point with smaller $y$-coordinate).
Lemma~\ref{lem:digital_path} shows that the way we picked the points would gives a digital path following the $4$-neighbor grid topology. 

The last step of the construction is as follows: For those points $(p_x,p_y) \in \domain/\{(0,0)\}$ not having an edge to $(p_x-1,p_y)$ or $(p_x,p_y-1)$, we connect $(p_x,p_y)$ to $(p_x-1,p_y)$ if $p_x \geq p_y$, otherwise to $(p_x, p_y-1)$. We call this construction GREEDY. 


The following lemma shows that every two consecutive points we picked in $L_j$ and $L_{j+1}$ to form the greedy digital path is connected under the $4$-neighbor topology.

\begin{lemma}\label{lem:digital_path}
	For $i=2,\ldots, \log_2 N$, any greedy digital path from $p \in L_{2^{i-1}}$ to $2p\in L_{2^i}$ is connected under the $4$-neighbor topology and it is $xy$-monotone.
\end{lemma}

\begin{proof}
	It is trivial for $p=(0,2^{i-1})$ and $p=(2^{i-1},0)$, so we ignore these two cases in the following.
	Suppose that there exists two consecutive points we picked in the greedy path $u \in L_j$ and $v \in L_{j+1}$ for some $2^{i-1} \leq j \leq 2^i-1$ such that $u$ and $v$ are not connected in the 4-neighbor topology, i.e., $u_x \neq v_x$ and $u_y \neq v_y$. Since $u$ and $v$ are grid points, this implies that $||u-v||_{\infty} \geq 2$.
	Let $u'$ and $v'$ be the points on the line segment $\overline{p,2p}$ with the smallest $|| \cdot ||_{\infty}$ distance to $u$ and $v$ respectively, i.e., $u'$ (resp. $v'$) is the intersection of $\overline{p,2p}$ and $x+y=j$ (resp. $x+y=j+1$). Since the slope of 
	$\overline{p,2p}$ is between $0$ and $\infty$ exclusively, $||u'-v'||_{\infty} < 1$.
	Furthermore, we know that the $|| \cdot ||_{\infty}$ distance between any two consecutive points on $L_j$ is 1, so $||u-u'||_{\infty} \leq 0.5$ and the same holds for $|| v-v' ||_{\infty}$. By triangle inequality, we have $||u-v||_{\infty} \leq ||u-u'||_{\infty} + ||u'-v'||_{\infty} + ||v'-v||_{\infty} <2$, which gives a contradiction.
	
	It is easy to see that the greedy digital path is $xy$-monotone because we only pick one grid point per each $L_j$. 
\end{proof}

Next, we show that any two greedy digital paths in the same slice $S_i$ are disjoint so that when we concatenate all the greedy digital paths slice by slice, it is easy to see that they form a tree rooted at the origin in Lemma~\ref{lem:tree}.

\begin{lemma}
	\label{lem:disjoint}
	For $i=2,\ldots, \log_2 N$, any two greedy digital paths in $S_i$ are disjoint.
\end{lemma}

\begin{proof}
	By the way we picked the grid points on the greedy digital paths, we can see that any grid point on the greedy digital path has at most 0.5 $L$-infinity distance to the corresponding line segment.
	Give any two consecutive line segments $\overline{p,2p}$ and $\overline{q,2q}$ in $S_i$ where $q_x = p_x +1$, for any point $v \in \overline{{p,2p}}/\{p\}$, the $|| \cdot ||_{\infty}$ distance from $v$ to $\overline{q,2q}$ is larger than $1$. Hence, one grid point cannot be assigned to more than one greedy digital path.
\end{proof}

Then, we show how the greedy digital paths in $S_i$ connect to the greedy digital paths in $S_{i-1}$, which gives us the structure of $dig(o,p)$ how it passes through some intermediate points.


\begin{lemma}\label{lem:connect}
	For $i=3,\ldots, \log_2 N$ and any $p = (p_x, p_y) \in L_{2^{i-1}}$, $dig(p, 2p) \subset dig(q, 2p)$, where $q = (\lfloor p_x/2\rfloor, \lceil p_y/2\rceil)$ if $p_x \geq p_y$, otherwise $q=(\lceil p_x/2\rceil, \lfloor p_y/2\rfloor)$.
\end{lemma}

\begin{proof}
	When $p_x$ is even, $q = (p_x/2, p_y/2) \in L_{2^{i-2}}$, so $dig(q,p)$ is a greedy digital path. $dig(q,2p)$ is the concatenation of $dig(q,p)$ and $dig(p,2p)$.
	Hence, $dig(p, 2p) \subset dig(q, 2p)$.
	
	When $p_x$ is odd, let $p_x = 2k + 1$ for some $k = 0,1,\ldots, 2^{i-2}-1$. Then, $p_y = (2^{i-1} - 2k - 1)$. Assume that $p_x \geq p_y$ (another case can be proved by the same approach). By the last step of GREEDY construction, $p$ will connect horizontally to $(p_x-1, p_y)$ and so on, so we follow the digital path from $p$ horizontally until we hit some greedy path. Since $p$ is between $p' = (2k,2^{i-1} - 2k)$ and $(2k+2,2^{i-1} - 2k - 2)$ on $L_{2^{i-1}}$, the greedy path we hit is $dig(p'/2,p')$ and then we follow $dig(p'/2,p')$ and reach $p'/2 = (k, 2^{i-2} - k) = (\lfloor p_x/2\rfloor, \lceil p_y/2\rceil)$. Therefore, $dig(p,2p) \subset dig(q,2p)$, where $q = p'/2$.
\end{proof}

By repeatedly applying Lemma~\ref{lem:connect} for all $i=3,\ldots, \log_2 N$, we have the following corollary.

\begin{cor}\label{cor:connect}
	For $i=3,\ldots, \log_2 N$ and any $p = (p_x, p_y) \in L_{2^{i-1}}$, $dig(o,p)$ passes through all the points $(\lfloor p_x/2^j\rfloor, \lceil p_y/2^j\rceil)$ if $p_x \geq p_y$, otherwise $(\lceil p_x/2^j\rceil, \lfloor p_y/2^j\rfloor)$ for $j=1,\ldots, i-2$.
\end{cor}

Now we have all the tools to show that GREEDY is a rooted tree at the origin so that GREEDY is a weak CDR.

\begin{lemma}
	\label{lem:tree}
	GREEDY is a rooted tree at the origin with $xy$-monotone paths to all the vertices.
\end{lemma}

\begin{proof}
	If we can show that every grid point $v \in \domain$ except the origin has exactly one edge to either $(v_x-1,v_y)$ or $(v_x,v_y-1)$, then GREEDY is a tree with $xy$-monotone paths connecting to all the grid points from the origin because the graph is connected in $\domain$ and there are $|\domain|-1$ edges. The $xy$-monotone is come from the fact that all the grid points are connected towards the origin. Clearly, this holds in $S_1$ because this part is a CDR with $N=2$. Hence, we consider $S_i$ for $i=2,\ldots,\log_2 N$.
	
	By Lemmas~\ref{lem:digital_path} and \ref{lem:disjoint}, we know that all the greedy digital paths in $S_i$ are $xy$-monotone and disjoint. Hence, for any point $v \in dig(p,2p)/\{p\}$ with $p \in L_{2^{i-1}}$, there is only one edge to either $(v_x-1,v_y)$ or $(v_x,v_y-1)$.
	
	Furthermore, the last step of GREEDY only applies to the grid points $p \in \domain/\{(0,0)\}$ not having an edge to $(p_x-1, p_y)$ or $(p_x,p_y-1)$. Hence, there is only one edge to $(p_x-1, p_y)$ or $(p_x,p_y-1)$ assigned to those points. 
\end{proof}

Then, we can talk about the quality of GREEDY in term of the number of inner leaves and the Hausdorff error in the next two lemmas.

\begin{lemma}
	\label{lem:death}
	There are at most $N^2/12$ inner leaves in GREEDY.
\end{lemma}

\begin{proof}
	By Corollary~\ref{cor:connect}, every grid point on the greedy digital paths can be extended to $L_N$. 
	Thus, in each $S_i$, the inner leaves are created between two greedy digital paths $dig(p,2p)$ and $dig(q,2q)$ exclusively, where $p_x+p_y = 2^{i-1}$, $q_x = p_x + 1$ and $q_y = p_y -1$. We only consider the grid points below the line $x=y$. The other case is symmetric. By GREEDY construction, all the grid points between $dig(p,2p)$ and $dig(q,2q)$ exclusively are connected to their left hand side neighbors. Hence, there is at most one inner leaf per each horizontal line between $dig(p,2p)$ and $dig(q,2q)$, except the line $y=2p_y-1$ because $(2p_x+1,2p_y-1)$ can be extended to $S_{i+1}$. Therefore, there are at most $\sum_{y=p_y}^{2p_y-2} 1 = p_y-1$ inner leaves between $dig(p,2p)$ and $dig(q,2q)$. By considering all $dig(p,2p)$ for $p_y = 1,\ldots, 2^{i-2}$ and the symmetric case, we have $2(\sum_{p_y=1}^{2^{i-2}} p_y -1 )= 2^{i-2}(2^{i-2}-1)$ inner leaves in $S_i$. We sum up for all the slices, we get $\sum_{i=2}^{\log_2 N}  2^{i-2}(2^{i-2}-1) = \frac{4^{\log_2 N-1}-1}{4-1} - \frac{2^{\log_2 N -1}-1}{2-1} = N^2/12 - 1/3 -N/2 +1 < N^2/12$ inner leaves. 
\end{proof}

\begin{lemma}
	\label{lem:error}
	The Hausdorff error in GREEDY is at most 5/2.
\end{lemma}

\begin{proof}
	Given any point $v \in \domain$, there exists some $n \leq N$ such that $2^{n-1} < v_x + v_y \leq 2^n$. Based on our slice partition, we partition $\overline{o,v}$ into $n$ line segments such that each line segment lie in some slice $S_i$. Then, we bound the Hausdorff error between each line segment and the corresponding piece of digital segment within each slice, which implies the overall Hausdorff error. Recall that $H(A,B)$ is the Hausdorff distance between $A$ and $B$ under $|| \cdot ||_{\infty}$ metric.
	
	We first show how to partition $\overline{o,v}$ and $dig(o,v)$. 
	For the sake of simplicity, we assume that $v_x \geq v_y$.
	Let $k= v_x + v_y$ and let $v_i$ be the intersection of $\overline{o,v}$ and $x+y= 2^i$ for $i \leq n-1$, i.e., $v_{i,x} = 2^{i}v_x/k$ and $v_{i,y} = 2^{i} v_y/k $.
	Let $v' = (\lfloor 2^{n-1}v_x/k \rfloor, \lceil 2^{n-1} v_y/k \rceil)$ and $v'' = (\lfloor 2^{n-1}v_x/k \rfloor + 1, \lceil 2^{n-1} v_y/k \rceil - 1)$ so that  $v_{n-1}$ is between $v'$ and $v''$ on line $x+y=2^{n-1}$. Therefore, $v$ is between $\overline{v',2v'}$ and $\overline{v'',2v''}$. Based on the GREEDY construction, $v$ is also between $dig(v',2v')$ and $dig(v'', 2v'')$.
	Suppose that $v$ is not in $dig(v'', 2v'')$, then $dig(o,v)$ is constructed by a horizontal path from $v$ to $dig(v',2v')$ and then following $dig(o,2v')$ to the origin. Let $v'_i = (\lfloor 2^{i}v_x/k\rfloor, \lceil 2^{i}v_y/k\rceil)$ for $i = 2, \ldots, n-1$. Since $dig(o,v)$ passes through $v'=v'_{n-1}$, by Corollary~\ref{cor:connect}, $dig(o,v)$ also passes through all $v'_i$. Then, we need to consider $H(\overline{v_{n-1}, v}, dig(v',v))$ and $H(\overline{v_{i-1}^{},v_i}, dig(v'_{i-1},v'_i))$, whose maximum gives the bound on $H(\overline{o,v}, dig(o,v))$.
	
	Now we are going to bound $H(\overline{v_{n-1}, v}, dig(v',v))$.
	Let $p$ (resp. $p'$) be the intersection of $\overline{v',2v'}$ (resp. $dig(v',2v')$) and $x+y=k$ and $q$ (resp. $q'$) be the intersection of $\overline{v'',2v''}$ (resp. $dig(v'',2v'')$) and $x+y=k$ so that $dig(v',v)$ is between $dig(v',p')$ and $dig(v'',q')$ (see \cref{fig:greedy}). Then, $H(\overline{v_{n-1}, v}, dig(v',v))$ is bounded by the maximum of $H(\overline{v_{n-1}, v}, dig(v',p'))$ and $H(\overline{v_{n-1}, v}, dig(v'',q'))$. Furthermore, $H(\overline{v_{n-1}, v}, dig(v',p')) \leq H(\overline{v_{n-1}, v}, \overline{v', p}) + H(\overline{v',p}, dig(v',p'))$ and $H(\overline{v_{n-1}, v}, dig(v'',q')) \leq H(\overline{v_{n-1}, v}, \overline{v'', q}) +$
	\newline $H(\overline{v'',q}, dig(v'',q'))$. 
	Because of the GREEDY construction, both $H(\overline{v',p}, dig(v',p'))$ and $H(\overline{v'',q}, dig(v'',q'))$ are at most $0.5$. Since $H(\overline{v', p}, \overline{v'', q}) < 2$, we have $H(\overline{v_{n-1}, v}, dig(v',v)) < 5/2$.
	
	For the second part, $H(\overline{v_{i-1}^{},v_i}, dig(v'_{i-1},v'_i))$ is bounded by 
	$H(\overline{v_{i-1},v_i}, \overline{v'_{i-1},v'_i}) + $ 
	\newline $H(\overline{v'_{i-1},v'_i}, dig(v'_{i-1},v'_i))$.
	Since $|| v_{i-1} - v'_{i-1}||_\infty$ and $|| v_{i} - v'_{i}||_\infty$ are at most 1, $H(\overline{v_{i-1},v_i}, \overline{v'_{i-1},v'_i})$ is also at most 1.
	Based on the GREEDY construction, we know that  $H(\overline{p,2p}, dig(p,2p)) \leq 0.5$ for any $p\in L_{2^i}$.
	If $v'_{i} = 2v'_{i-1}$, we are done. Otherwise, $v'_i = (2v'_{i-1,x}+1, 2v'_{i-1,y}-1 ) $, then $dig(v'_{i-1},v'_i)$ is constructed by a horizontal path from $v'_i$ to $dig(v'_{i-1},2v'_{i-1})$, and then following $dig(v'_{i-1},2v'_{i-1})$ to $v'_{i-1}$. Using the similar argument, we can show that $H(\overline{v'_{i-1},v'_i}, dig(v'_{i-1},v'_i)) \leq 1.5$. Overall, it gives $H(\overline{o,v}, dig(o,v)) \leq 5/2$ when $dig(o,v)$ passes through $v'$.
	
	We go back to another case when $v \in dig(v'',2v'')$. Let $v''_i = (\lfloor 2^{i}v_x/k + 1/2^{n-1-i}\rfloor, \lceil 2^iv_y/k - 1/2^{n-1-i} \rceil)$ for $i=2,\ldots, n-1$. Since $dig(o,v)$ passes through $v''=v''_{n-1}$, by Corollary~\ref{cor:connect}, $dig(o,v)$ also passes through all $v''_i$. Since $||v_{i} - v''_{i}||_\infty$ is at most 1 for $i=2,\ldots,n-1$, we can apply the same argument as above to show that $H(\overline{o,v}, dig(o,v)) \leq 5/2$ when $dig(o,v)$ passes through $v''$.
\end{proof}

Combining all these lemmas, we have our main theorem.


\begin{restatable}{theorem}{Greedy}
	\label{theo_greedy}
	For any $N>0$ we can create a weak CDR in $\Z^2$ with $5/2$ error and $N^2/12$ inner leaves. 
\end{restatable}

\begin{proof}
	Lemma~\ref{lem:tree} guarantees that GREEDY satisfies axioms 1, 2, 3 and 5. Lemmas~\ref{lem:death} and \ref{lem:error} give the two qualities of the weak CDR.
\end{proof}

From the above theorem, we can also extend it to have a trade-off construction with $O(c)$ error and $O(N^2/c)$ by scaling the tree.

\begin{theorem}
	\label{thm:tradeoff}
	For any $N>0$ and $c \leq N$ we can create a weak CDR in $\Z^2$ with $O(c)$ error and $O(N^2/c)$ inner leaves.
\end{theorem}

\begin{figure}
	\centering
	\includegraphics[width=.8\textwidth]{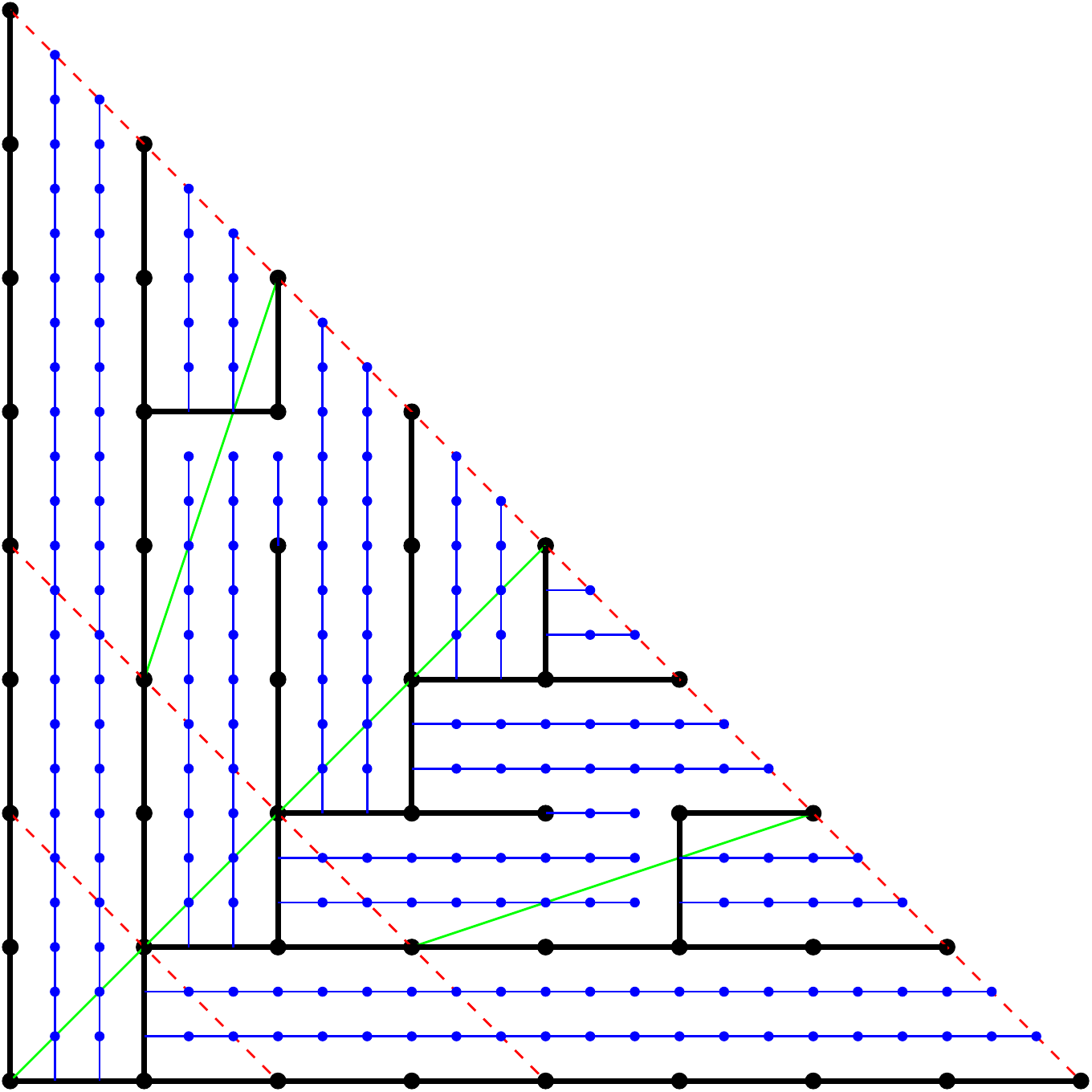}
	\caption{Example of a weak CDR with $O(c)$ error and $\Theta(N^2/c)$ inner leaves for $N=24$ and $c=3$. The thick black line segments represent the GREEDY construction for $N/c=8$. The blue segments cover the remaining refined grid vertices.}
	\label{fig:trade-off}
\end{figure}

\begin{proof}
	We first apply the GREEDY construction in $\mathcal{G}^+_{\lceil N/c \rceil}$, in which we have $O(1)$ error and $O(N^2/c^2)$ inner leaves. Then, we scale up the tree by a factor of $c$ so that the original grid edges have length $c$. Therefore, the error of the tree in $\mathcal{G}^+_{\lceil N/c \rceil \cdot c}$ becomes $O(c)$, but the tree does not cover all the refined grid vertices. Then, we draw some vertical or horizontal line segments that branch from the tree as shown in Figure~\ref{fig:trade-off}. This will increase the number of inner leaves by a factor of $c$, i.e., $O(N^2/c)$. Each new branch is a copy of some sub-path of the original GREEDY tree and is shifted by at most $c$ steps. Hence, their errors are still $O(c)$. 
\end{proof}

In addition to the GREEDY construction for a weak CDR, we can also observe some nice properties in GREEDY. If we remove all the branches which end up at some inner leaves, we have an infinite tree which covers more than half of the grid points of $\domain$. In particular, for each grid point not in the tree, there must exist a vertex in the tree within one unit distance. Hence, given any point $p$ in $\domain$, we can snap $p$ to some vertex $q$ in the tree with distance at most $1$. Then, $dig(o,p)$ can be approximated by $dig(o,q)$ with a very small distortion at the end point, but the Hausdorff error $H(\overline{o,p},dig(o,q))$ is $O(1)$. 
Let $T_G$ be the tree created by GREEDY after removing all the inner branches. 
We give a formal statement as follows:

\begin{theorem}
	Let $V$ be the vertices in $T_G$. Then, $DS(\{o\} \times V)$ realized by $T_G$ is a partial CDS with $O(1)$ error.	
	Moreover, for any vertex $p \in \Z^2$, there exists a vertex $q \in V$ such that $||p-q||_\infty\leq 1$. 
\end{theorem}

\begin{proof}
	By Lemma~\ref{lem:tree}, we know that every path from any vertex in $T_G$ to the origin is $xy$-monotone.
	And by Corollary~\ref{cor:connect}, we know that all the greedy digital paths can be extended to infinity.
	Hence, $DS(\{o\} \times V)$ realized by $T_G$ satisfies all the five axioms. Lemma~\ref{lem:error} gives the $O(1)$ error.
	
	
	By definition of greedy digital path, in each $L_j$ between two consecutive greedy digital paths there is at most one point not on these two paths. Hence, the distance between those points and greedy paths is at most one.
\end{proof}



\section{Point sets with constant discrepancy} \label{GreedyDiscrepancy}

In this section we construct red $R$ and blue $B$ point sets such that the absolute value of their discrepancy is 1. Let $m > 0$ be the difference between the number of blue and red points as defined in \cref{DiscrepancySection}. Our construction has $\Theta(m^2)$ many points. Afterwards we also prove that a discrepancy of 1 cannot be achieved with $o(m^2)$ many points.

We first describe a specific configuration of points, called staircase.
\begin{definition}
	A staircase is a sequence of alternating blue and red points $(p_1,p_2,...,p_n)$ in the unit square. It starts and ends with a blue point. Moreover for every red point $p_i$, the blue point $p_{i-1}$ has smaller $x$-coordinate and the same $y$-coordinate. The blue point $p_{i+1}$ has the same $x$-coordinate and smaller $y$-coordinate. 
\end{definition}
\begin{figure}
	\centering
	\includegraphics[width = 0.4\textwidth,trim= 50 150 50 150, clip]{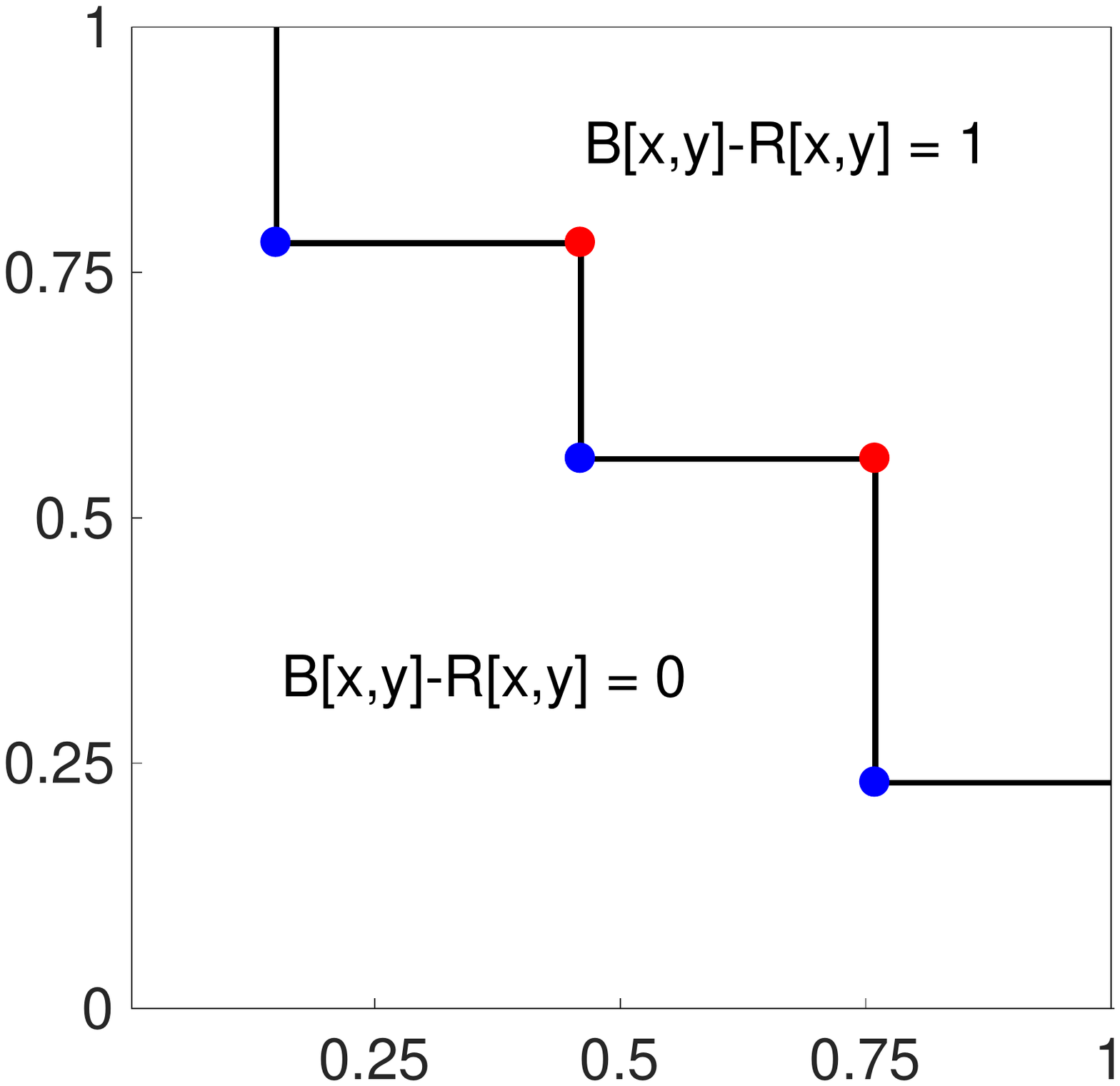}
	\caption{A Staircase}\label{Fig:Staircase}
\end{figure}

Given a staircase, we can define a curve by connecting consecutive points on the staircase. Additionally we add a vertical segment at the beginning and a horizontal segment at the end, in order to connect the curve to the boundary of the unit square, see \cref{Fig:Staircase}. We will also use the term ``staircase" for this curve. 

\begin{observation}
	The transformation in \cref{TransformationSection} maps a CDR to blue and red points in the unit square, which can be decomposed into a set of staircases.
\end{observation}
\begin{proof}
	Below every red point there is a corresponding blue point. Moreover, by Lemma~\ref{lem_one_line}, in each row the set of blue and red points is alternating, i.e., to the left of every red point there is a corresponding blue point.
\end{proof}

Assume that a set of blue and red points forms one staircase. Then the curve induced by the staircase splits the unit square into two parts. The set of points $(x,y)$ to the bottom-left of the staircase satisfies $B[x,y]-R[x,y] = 0$ whereas the set of points to the top-right of the staircase satisfies $B[x,y]-R[x,y] = 1$, see \cref{Fig:Staircase}. When the set of blue and red points can be decomposed into many staircases, then we can easily compute the value $B[x,y]-R[x,y]$ by counting how many staircases are to the bottom-left of point $(x,y)$.

Recall the definition of discrepancy of $R$ and $B$ at point $(x,y)$:
\begin{equation*}
	D_{R,B}(x,y) = m x y - (B[x,y]-R[x,y]).
\end{equation*}
The first term $m x y$ represents the expected difference between the numbers of blue and red points in the axis-aligned rectangle with corner points $(0,0)$ and $(x,y)$. Every point $(x,y)$ along the curves $C_i := \{(x,y) \in [0,1]^2 | x\cdot y = \frac{i}{m}\}$, where $i \in \{0,1,...,m\}$, describes a rectangle $[0,x] \times [0,y]$ in which we expect $i$ many blue points more than red points. \cref{Fig:StaircaseLowerBound} illustrates the curves $C_i$ in black for $m=7$.

The idea of our construction is to approximate the level curves $C_{i-0.5}$ by staircases, where $i \in \{1,...,m\}$. We will construct $m$ staircases such that the staircase approximating $C_{i-0.5}$ is between $C_{i-1}$ and $C_{i}$. This guarantees that the discrepancy $D^*_{R,B}$ is at most 1.

We describe how we construct the staircase which approximates $C_{i-0.5}$. We start with a blue point at the intersection of the two curves $C_{i-1}$ and $x = y$. This is the blue point $(\frac{\sqrt{i-1}}{\sqrt{m}},\frac{\sqrt{i-1}}{\sqrt{m}})$. Starting from there we move horizontally to the right until we hit the curve $C_i$ at the point $(\frac{i}{\sqrt{i-1} \sqrt{m}}, \frac{\sqrt{i-1}}{\sqrt{m}})$. We add a red point here. Then we move vertically down until we hit $C_{i-1}$ and put a blue point. We continue in this fashion, i.e. from a blue point on $C_{i-1}$ we move horizontally to the right and put a red point on $C_i$. From a red point on $C_i$ we move vertically down and put a blue point on $C_{i-1}$. 
The blue points will have the coordinates $(\frac{i^k}{(i-1)^{k-0.5} \cdot \sqrt{m}}, \frac{(i-1)^{k+0.5}}{i^k \cdot \sqrt{m}})$ and the red points have the coordinates $(\frac{i^{k+1}}{(i-1)^{k+0.5} \cdot \sqrt{m}}, \frac{(i-1)^{k+0.5}}{i^k \cdot \sqrt{m}})$, where $k \in \{0,1,2,...\}$.
We stop this construction when we leave the unit square, i.e., we look for the largest $k$ such that the blue point $(\frac{i^k}{(i-1)^{k-0.5} \cdot \sqrt{m}}, \frac{(i-1)^{k+0.5}}{i^k \cdot \sqrt{m}})$ is still contained in $[0,1]^2$.
The maximum value for $k$ is
\begin{equation*}
	k = \left \lfloor \frac{\log\left(\frac{\sqrt{m}}{\sqrt{i-1}}\right)}{\log\left(\frac{i}{i-1}\right)} \right \rfloor
\end{equation*}
for $i \geq 2$ and $k=0$ for $i=1$.
So far we described how we construct the staircases on the side $y \leq x$. We add red and blue points on the side $y > x$ to make the construction symmetric to the line $y = x$. \cref{Fig:StaircaseLowerBound} illustrates our construction, which we call the symmetric greedy staircase construction.

\begin{figure}
	\centering
	\includegraphics[width = 0.7\textwidth,trim= 0 0 0 0, clip]{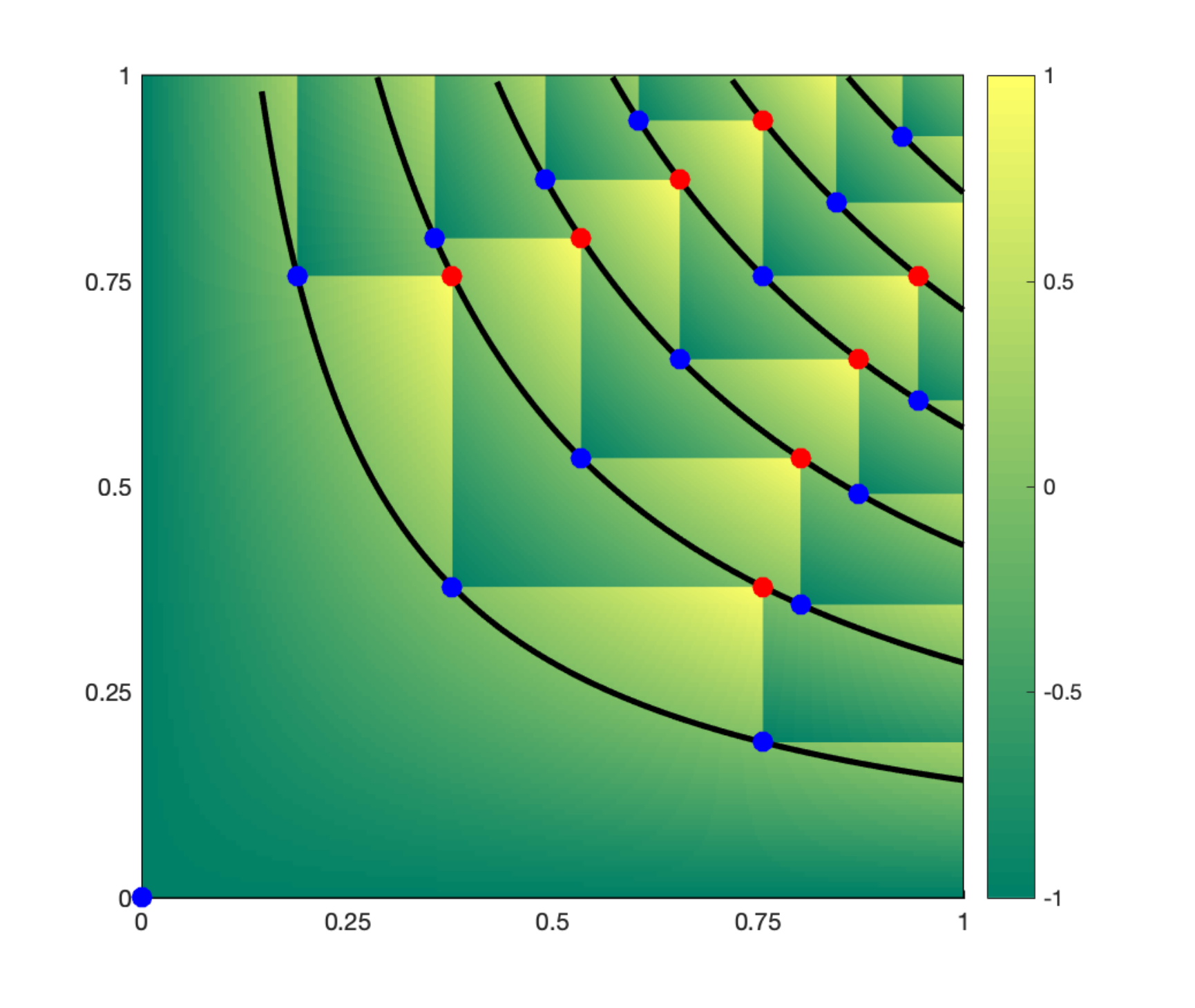}
	\caption{Staircase approximation for $m=7$. \\
		The curves $C_i$ are drawn in black. The brightness of the green color encodes the value of the discrepancy $D_{R,B}(x,y)$ at each point $(x,y)$. The discrepancy values range between -1 and 1 as shown on the right hand side. The staircases can be seen at the discontinuity of the discrepancy function. At each staircase the discrepancy function changes its value by 1.
	}\label{Fig:StaircaseLowerBound}
\end{figure}

\begin{observation}
	The points of the symmetric greedy staircase construction with $m$ stairs are
	\begin{align*}
		B & = \left\{\left(\frac{i^k}{(i-1)^{k-0.5} \cdot \sqrt{m}}, \frac{(i-1)^{k+0.5}}{i^k \cdot \sqrt{m}} \right) \Big| i \in \{1,2,...,m\} \text{ and } -k^*(i) \leq k \leq k^*(i) \right\} \\
		R & = \left\{\left(\frac{i^{k+1}}{(i-1)^{k+0.5} \cdot \sqrt{m}}, \frac{(i-1)^{k+0.5}}{i^k \cdot \sqrt{m}} \right) \Big| i \in \{1,2,...,m\} \text{ and } -k^*(i) \leq k \leq k^*(i)-1 \right\}
	\end{align*}
	where
	\begin{equation*}
		k^*(i) = \begin{cases}
			0 & \text{if } i = 1 \\
			\left \lfloor \frac{\log\left(\frac{\sqrt{m}}{\sqrt{i-1}}\right)}{\log\left(\frac{i}{i-1}\right)} \right \rfloor & \text{if } i \neq 1.
		\end{cases} 
	\end{equation*}
	There are $2 k^*(i)+1$ (resp. $2k^*(i)$) many blue (resp. red) points in the $i$-th staircase.
\end{observation}

\begin{theorem}
	The symmetric greedy staircase construction has discrepancy 1.
\end{theorem}
\begin{proof}
	Consider any point $(x,y)$ between the $i$-th and $(i+1)$-th staircase. It holds that $B[x,y]-R[x,y] = i$. Moreover both staircases are bounded from below by the level curve $C_{i-1}$ and from above by $C_{i+1}$, which means that $\frac{i-1}{m} \leq x \cdot y \leq \frac{i+1}{m}$. Summarizing, we can bound the discrepancy 
	\begin{equation*}
		-1 \leq \underbrace{m x y - (B[x,y]-R[x,y])}_{= D_{R,B}(x,y)} \leq 1.
	\end{equation*}
\end{proof}

\begin{theorem} \label{upperboundnumber}
	The symmetric greedy staircase construction with $m$ stairs has $O(m^2)$ many points.
\end{theorem}
\begin{proof}
	The number of blue points, which are used in our construction, is
	\begin{align*}
		\abs{B} & = \sum_{i = 1}^{m} 1+2\cdot k^*(i) = m + \sum_{i = 2}^{m} 2\cdot \left \lfloor \frac{\log\left(\frac{\sqrt{m}}{\sqrt{i-1}}\right)}{\log\left(\frac{i}{i-1}\right)} \right \rfloor 
		\leq  O(m) +  \sum_{i = 2}^{m-1} \frac{\log\left(\frac{m}{i}\right)}{\log\left(\frac{i+1}{i}\right)}.
	\end{align*}
	We now lower bound the denominator by
	\begin{equation*}
		\log\left(\frac{i+1}{i}\right) =\log\left(1+\frac{1}{i}\right) \geq \frac{1}{i} - \frac{1}{i^2} = \frac{i-1}{i^2}.
	\end{equation*}
	Putting the inequalities together, we get:
	\begin{align*}
		\abs{B} & \leq O(m) + \sum_{i = 2}^{m-1} \frac{i^2}{i-1} \log\left(\frac{m}{i}\right) 
		\leq O(m) + 2 \sum_{i = 2}^{m-1} i \log\left(\frac{m}{i}\right)
	\end{align*}
	The continuous function $f(i) = i \log \left(\frac{m}{i}\right)$ has exactly one maximum in the interval $[2,m]$ with a value bounded by $m \log m$ and is monotone on both sides of it. Therefore we can replace the sum by an integral.
	\begin{align*}
		\abs{B} & \leq O(m \log m) + 2 \int_{2}^{m-1} i \log\left(\frac{m}{i}\right) di = O(m \log m) + 2  \left[ \frac{i^2}{4} \cdot \left(1+2 \log\left(\frac{m}{i}\right)\right) \right]_{i=2}^{i=m-1} \\
		& = O(m^2).
	\end{align*}
\end{proof}

We now show that our construction is tight.

\newcommand{\const}{\xi}

\begin{theorem} \label{Non-intersecting staircases}
	Let $B$ and $R$ be point sets, which can be decomposed into $m$ non-intersecting staircases, and have a discrepancy bounded by a constant $\const$. Then $\abs{B} = \Omega\left(m^2\right)$.
\end{theorem}
\begin{proof}
	The $i$-th staircase is bounded from below by the level curve $C_{i-\const}$ and from above by $C_{i+\const-1}$ because of the discrepancy constraint. We count how many points are necessary to create the $i$-th stair. The minimum number can be realized by constructing a stair in a greedy manner between $C_{i-\const}$ and $C_{i+\const-1}$ because both curves are convex.	
	\begin{align*}
		B & = \left\{\left(\frac{(i+\const-1)^k}{m (i-\const)^{k-1} }, \frac{(i-\const)^k}{(i+\const-1)^k} \right) \Big| i \in \{1,2,...,m\} \text{ and } 1 \leq k \leq k^*(i) \right\} \\
		R & = \left\{\left(\frac{(i+\const-1)^{k+1}}{m (i-\const)^k }, \frac{(i-\const)^k}{(i+\const-1)^k} \right) \Big| i \in \{1,2,...,m\} \text{ and } 1 \leq k \leq k^*(i)-1 \right\}
	\end{align*}
	where
	\begin{equation*}
		k^*(i) = \begin{cases}
			0 & \text{if } i \leq \const \\
			\left \lfloor \frac{\log\left(\frac{m}{i-\const}\right)}{\log\left(\frac{i+\const-1}{i-\const}\right)} \right \rfloor & \text{if } i > \const.
		\end{cases} 
	\end{equation*}
	The number of blue points can therefore be bounded by		
	\begin{align*}
		\abs{B} & \geq \sum_{i = 1}^{m} k^*(i) 
		\geq \const + \sum_{i = \const+1}^{m} \left\lfloor \frac{\log\left(\frac{m}{i-\const}\right)}{\log\left(\frac{i+\const-1}{i-1}\right)} \right\rfloor 
		\geq -O(m) +  \sum_{i = 1}^{m-\const} \frac{\log(\frac{m}{i})}{\log(\frac{i+2\const-1}{i})}
	\end{align*}
	Using the inequality $\log\left(\frac{i+2\const-1}{i}\right) =\log(1+\frac{2\const-1}{i}) \leq \frac{2\const-1}{i}$ and comparing the sum with an integral, as done in the proof of \cref{upperboundnumber}
	\begin{equation*}
		\frac{1}{2\const-1} \sum_{i = 1}^{m-\const} i \log\left(\frac{m}{i}\right) \geq \frac{1}{2\const-1} \left(-O(m \log m) + \int_{1}^{m-\const} i \log\left(\frac{m}{i}\right) di \right)
	\end{equation*}
	we can conclude $\abs{B} \geq \Omega\left(m^2\right)$.
\end{proof}


\begin{restatable}{theorem}{DiscrepOne}
	\label{TheoremDiscrepancy1}
	Let $B$ and $R$ be two sets of points whose discrepancy satisfies $D_{R,B}^*<1$. Then $\abs{B} = \Omega(m^2)$, where $m = \abs{B}-\abs{R}$.
\end{restatable}

\begin{proof}
	Consider the sets $S_i := \{(x,y) \in [0,1]^2 \,\Big|\, B[x,y]-R[x,y] = i\}$ for $i \in \{0,1,...,m\}$. Because the discrepancy of the point set $B$ and $R$ is less than 1 we can conclude that
	\begin{enumerate}
		\item the curves $C_i$ are contained in $S_i$ and
		\item the points between $C_i$ and $C_{i+1}$ are either contained in $S_i$ or $S_{i+1}$.
	\end{enumerate}
	Therefore there exists a curve between $C_i$ and $C_{i+1}$ which is only neighboring $S_i$ to its bottom-left and $S_{i+1}$ to its top-right for each $i \in \{0,1,...,m-1\}$. This curve is a staircase. Hence there exists a staircase between $C_i$ and $C_{i+1}$ for each $i \in \{0,1,...,m-1\}$. Those staircases are non-intersecting because $D_{R,B}^*<1$. Therefore they consist of at least $\Omega(m^2)$ many points, as shown in \cref{Non-intersecting staircases}.
\end{proof}

As mentioned before, our transformation maps a CDR to a set of points in the unit square, which can be decomposed into staircases. Unfortunately \cref{Non-intersecting staircases} does not imply that every 2D weak CDR with constant discrepancy needs $\Omega(N^2)$ many leaves. The staircases from our transformation might be intersecting. Intersecting staircases can again be decomposed into non-intersecting (only touching) staircases, where there does not need to be a blue or red point at every turn.


\section{Final remarks}\label{sec_conclu}
Common intuition would say that the $\Omega(\log N)$ lower bound for the error of two-dimensional CDR and CDS automatically extends to higher dimensions.
The observation that this is not true opens up new ways in which research can continue. We believe that further analysis of the mapping between the three spaces (from CDR in high dimensions to the 2-D weak CDR to the two-colored pointset) and the high interdependence between the three spaces can help in designing better lower and upper bounds.

Our lower bound $\Omega(\log^{1/(d-1)} N)$ extends the previous lower bound. The next step would be to close the gap between $\Omega(\log^{1/2} N)$ and $O(\log N)$ bounds in three dimensions. Even if the final answer ends up being $\Theta(\log N)$ we believe that the relationship between high dimensional CDRs, weak CDRs induced in subspaces and the mapping to pointset gives a better understanding of CDRs. 

We also find that weak CDRs are an interesting research topic on their own. In particular, we would like to find the relationship between the number of inner leaves and the error of the construction. That is, say that we want a CDR with $O(e)$ error (for some $e\leq \log n$). What is the minimum number of leaves $\ell=\ell(e)$ that such a CDR must have? Can we find such a construction?

\cref{theo_lower_2} seems to indicate a linear relationship between the two, and it is not hard to obtain one (See, for example~\cref{thm:tradeoff} and \cref{fig:trade-off}). 
However, this construction is most likely not the best possible one. Indeed, even if we are interested in $O(\log N)$ error, this construction creates a large number of inner leaves, but we know of CDRs with the same error and no inner leaves. Thus, the question becomes, can we significantly improve upon the greedy construction in \cref{GreedyCDR}? Or is there some exponential dependency between the number of inner leaves and the error of the weak CDR?

\paragraph{Acknowledgements}
The authors would like to thank 
Matthew Gibson,
Evanthia Papadopoulou,
Andr\'e van Renssen and
Marcel Roeloffzen
for their helpful discussions during the creation of this paper.
The authors would also like to thank the anonymous reviewers for the many comments that helped improve the paper. We would especially like to thank a reviewer for SODA that showed us how to improve the lower bound from $\Omega(\log^{1/d} N)$ to $\Omega(\log^{1/(d-1)}N)$.

	\bibliography{reference}

\end{document}